\documentclass{IEEEtran}

\usepackage{amsmath, amsfonts, amsthm}
\usepackage{algorithm}
\usepackage{array}
\usepackage[caption=false,font=normalsize,labelfont=sf,textfont=sf]{subfig}
\usepackage{textcomp}
\usepackage{stfloats}
\usepackage{url}
\usepackage{verbatim}
\usepackage{graphicx}
\usepackage{cite}
\usepackage{booktabs}

\usepackage{xcolor}
\usepackage{amssymb}
\usepackage{mathtools}
\usepackage{algpseudocode}
\usepackage{tabularx}
\usepackage{xurl}
\usepackage{hyperref}

\newtheoremstyle{mytheorem}
{3pt}                
{3pt}                
{}        
{}                
{\bfseries}       
{.}               
{ }               
{}                
\theoremstyle{mytheorem}
\newtheorem{thm}{Theorem}

\theoremstyle{definition}
\newtheorem{define}{Definition}
\newtheorem{assume}{Assumption}
\theoremstyle{remark}

\newcommand{\R}{\ensuremath{\mathbb{R}}}

\def\T{\top}
\def\bmat{\begin{bmatrix}}
\def\emat{\end{bmatrix}}

\newcommand{\para}[1]{\left(#1\right)}		

\title{Distributed Traffic State Estimation in Connected Vehicle and Roadside Infrastructure Networks}

\author{Vincent de Heij, M. Umar B. Niazi, Saeed Ahmed,  Karl H. Johansson

\thanks{This work was supported by NXTGEN HighTech growthfund: Autonomous Factory, the Swedish Research Council's Distinguished Professor Grant, the Knut and Alice Wallenberg Foundation's Wallenberg Scholar Grant,  Digital Futures' Summer Research Internship Programme, and the Holland High Tech (TKI HTSM) strategic program PPS-I Flex HighTech under the project number 24PPS173-CABS.}
\thanks{Vincent de Heij and Saeed Ahmed are with the Engineering and Technology Institute Groningen, Faculty of Science and Engineering, University of Groningen, 9747 AG Groningen, The Netherlands (emails: v.de.heij@rug.nl, s.ahmed@rug.nl).}
\thanks{M. Umar B. Niazi and Karl H. Johansson are with the Department of Decision and Control Systems and with Digital Futures, KTH Royal Institute of Technology, SE-100 44 Stockholm, Sweden (emails: mubniazi@kth.se, kallej@kth.se).
}%
}

\begin{document}

\maketitle

\begin{abstract}
This paper proposes a distributed traffic state estimation framework that combines
infrastructure sensors and connected vehicles as cooperative
sensing nodes. Using Vehicle-to-Everything (V2X) communication, nearby nodes exchange local estimates and update them through a distributed Kalman filter designed for a second-order macroscopic traffic flow model. A consensus step fuses heterogeneous information across the network, while projection steps enforce physically consistent traffic states.
We evaluate the method on HighD and NGSIM data, and on microscopic SUMO simulations that capture transient congestion. The results show accurate reconstruction of highway traffic states and detection of nonlinear shockwave dynamics, even with sparse infrastructure sensing and intermittent vehicular connectivity. A statistical analysis further shows how CV penetration rate, V2X communication range, and infrastructure deployment affect estimation accuracy. In particular, with 10\% CV penetration, V2X ranges of 300-400 m, and sparse infrastructure deployment, the combined infrastructure-vehicle configuration consistently outperforms approaches that rely only on infrastructure or only on connected vehicles.

\end{abstract}


\section{Introduction}\label{sec:introduction}
Modern intelligent transportation systems, especially those using cooperative driving strategies for connected and automated vehicles \cite{lee2025}, require accurate real-time macroscopic traffic states such as density, velocity, and flow.
Because deploying fixed roadside sensor units (RSUs) to satisfy observability requirements is expensive, traffic state estimation (TSE) reconstructs these states by fusing data from existing sensors. RSUs, such as induction loops, radars, and cameras, provide accurate but sparse and stationary measurements. Connected vehicles (CVs), equipped with onboard sensors such as GPS, cameras, and LiDAR, provide local microscopic observations. Combining both sources enables multi-scale monitoring of traffic conditions.

Until recently, TSE has relied mainly on fixed infrastructure, with mobile data supplied only intermittently by probe vehicles that report aggregated trajectories rather than real-time measurements \cite{Keimer2018}. 
Most existing methods also use a centralized fusion architecture, where a transportation management center acts as a global estimator. This approach does not scale well and ignores key communication constraints in V2X networks. V2X protocols such as DSRC and C-V2X support only short-range direct communication, typically over a few hundred meters
\cite{ansari2021, zadobrischi2024, mendes2025,maglogiannis2022}.
As CV penetration increases \cite{statista2025},
the resulting data volume can exceed available bandwidth, causing latency and high computational cost \cite{arthurs2021}.
These limits make centralized estimation unsuitable for real-time control and create a single point of failure.

These issues motivate a distributed TSE (DTSE) architecture in which RSUs and CVs estimate traffic states locally by combining their own measurements with information received from nearby V2X nodes. This design matches the local nature of V2X connectivity and reduces the need to transmit raw sensor data to a central TMC. DTSE must still handle heterogeneous measurements, communication limits, and time-varying network topology caused by vehicle motion. The main challenge is therefore not just measurement sparsity, but also processing dense local data streams under strict latency constraints, especially in dynamic, non-equilibrium traffic.

Figure~\ref{fig:scenario} illustrates the value of the distributed architecture. A downstream bottleneck creates congestion and a backward-moving shockwave as the ego-CV approaches from upstream. Because the congestion is outside the vehicle’s sensing range, its local sensors cannot detect it. If DTSE provides an accurate global state estimate, the ego-CV can anticipate the shockwave and decelerate before reaching the bottleneck. This helps avoid adding to the congestion and allows the vehicle to act as a control input that helps dissipate the wave~\cite{lee2025}.

\begin{figure}[!t]
    \centering
    \includegraphics[width=\linewidth]{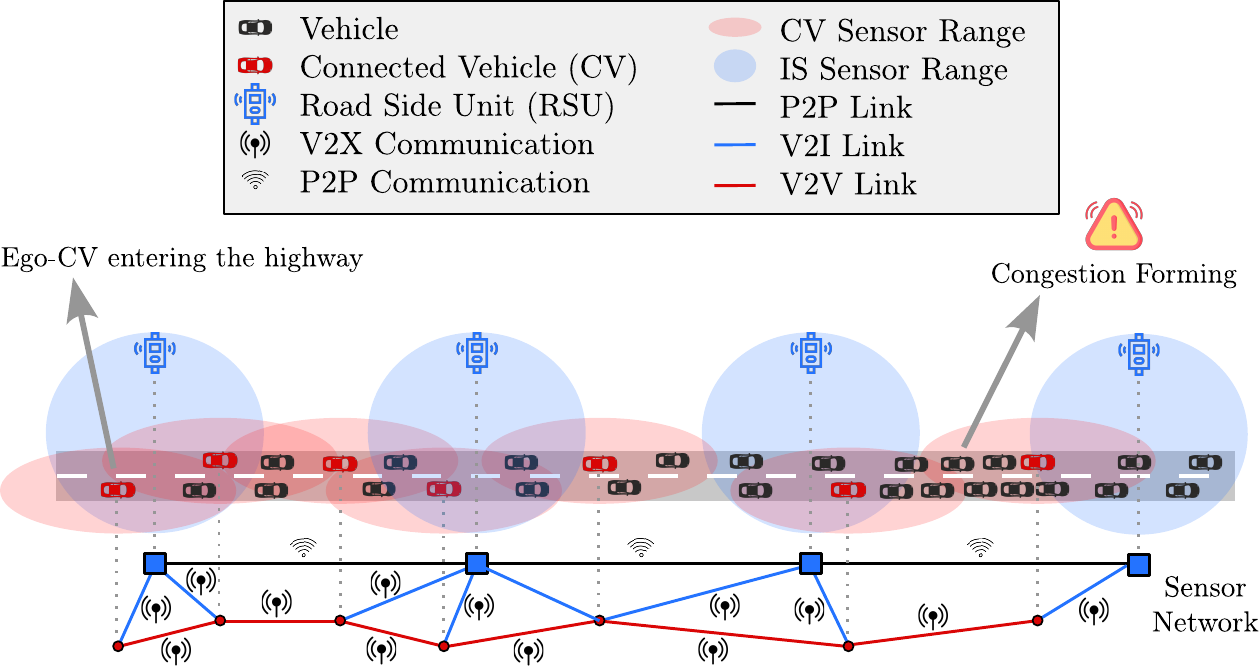}
    \caption{Motivating DTSE scenario on a V2X-enabled highway: an ego-CV upstream of a bottleneck cannot observe the resulting backward-propagating shockwave from its onboard sensors alone, but can reconstruct it by fusing local measurements with estimates from nearby RSUs and CVs.}
    \label{fig:scenario}
\end{figure}

\subsection{Related Work}
TSE methods are commonly grouped into model-based, data-driven, and hybrid approaches~\cite{seo2017review,di2021}. Since the present work is model-based, we first review methods that use macroscopic traffic-flow models to propagate traffic states over space and time. First-order models, such as the Lighthill-Whitham-Richards (LWR) model~\cite{lighthill1955,richards1956} and the cell transmission model (CTM)~\cite{daganzo1994}, describe traffic dynamics through a single conservation law for vehicle density and have been widely adopted in estimation frameworks, including ensemble and extended Kalman filtering, Lagrangian probe-data assimilation, and scalable localized estimation schemes~\cite{Work2008,Herrera2010,hinsbergen2012,Yuan2012}. However, their equilibrium speed-density assumption limits their ability to capture non-equilibrium traffic dynamics.

To address this limitation, second-order models, such as the Payne-Whitham (PW) model~\cite{payne1971,whitham1999} and its network-oriented implementations, including METANET~\cite{Kotsialos}, introduce an additional dynamic that accounts for driver anticipation and deviations from equilibrium, making them more suitable for representing non-equilibrium phenomena such as stop-and-go waves and shockwave propagation. However, the PW model exhibits physical inconsistencies under certain traffic conditions~\cite{daganzo1995}, which the Aw-Rascle-Zhang (ARZ) model~\cite{aw2000,zhang2002} was introduced to overcome. Owing to this physical consistency, the ARZ model has become the preferred second-order basis for TSE: heterogeneous data fusion was first addressed with an extended Kalman filter~\cite{seo2017}, later extended to a nonlinear ARZ state-space formulation with on- and off-ramp dynamics in which heterogeneous RSU and CV measurements are fused via moving horizon estimation~\cite{vishnoi2024}. Related second-order approaches include particle filtering with generalized ARZ-type dynamics for mixed human-driven and automated traffic~\cite{Wang2017} and boundary-observer design for the ARZ model~\cite{Yu2021}.



Most model-based TSE methods above are formulated around fixed infrastructure measurements, with mobile observations entering mainly as Lagrangian probe data. Since RSUs provide sparse and stationary views of traffic, several works use probe vehicles or CVs to provide measurements that move with the traffic stream. Existing approaches estimate traffic states and the fundamental diagram (FD) parameters from probe-vehicle spacing measurements~\cite{Seo2015}, recover density and flow in mixed traffic from CV speed reports and a small number of fixed flow measurements~\cite{bekiaris2016}, infer density and level of service from CV-generated distance headway, speed, and stop-count data~\cite{Khan17}, and use extended floating-car data in Bayesian TSE~\cite{kyriacou2023}. Together with the ARZ-based heterogeneous-fusion frameworks discussed above~\cite{seo2017,Wang2017,vishnoi2024}, these studies show that mobile sensing can enrich TSE beyond fixed-location measurements alone.

Despite this progress, many existing methods remain centralized, with measurements collected at a TMC before the traffic state is reconstructed. This limits scalability in V2X-enabled traffic systems, where sensing and communication are localized, mobile, and bandwidth-constrained. DTSE has therefore emerged as an alternative, with existing approaches including distributed switched observers, distributed Kalman filters, distributed state observers, and boundary observers for linearized ARZ models ~\cite{vivas2015,sun2014,sun2017,guo2017,zhang2020}. Among these methods, distributed Kalman filtering is especially relevant because it retains the recursive, uncertainty-aware structure of centralized Kalman filtering while replacing global measurement fusion with local prediction, correction, and neighbor-to-neighbor information exchange. 
However, most distributed TSE studies focus on fixed P2P networks between infrastructure sensors. Mixed RSU-CV networks, which introduce mobile sensing and estimation nodes and induce a time-varying V2X communication topology, have received comparatively less attention.

Data-driven TSE methods bypass explicit traffic-flow modeling and instead learn mappings from observed measurements to unobserved traffic states, often by exploiting recurring spatiotemporal patterns in historical data. Examples include detector-error diagnosis, loop-data imputation, tensor-based completion, and section-level $k$-nearest-neighbor imputation~\cite{Chen2003,Ran2016,Tak2016}. More recent hybrid learning methods combine data-driven estimation with traffic-flow physics, for example through physics-informed neural networks based on LWR or ARZ dynamics and neural-operator-based observers with sparse RSU measurements~\cite{Shi2022,Zhang2024,giraldo2026,barreau2021,Harting2025}. 
However, these learning-based methods rely on representative training data and typically require centralized training or global access to measurements,
limiting their applicability to online estimation over local, time-varying communication links.

\subsection{Our Contributions}
To address the challenge of online traffic state estimation in V2X networks of CVs and infrastructure, where sensing and communication are local, heterogeneous, and time-varying, we propose a DTSE framework. The method combines a distributed Kalman filter tailored to second-order traffic-flow dynamics to reconstruct density and velocity fields from local measurements and V2X exchange.

The main contributions of this paper are threefold. 
First, we introduce a decentralized traffic‑state estimation architecture in which both fixed RSUs and mobile CVs act as active sensing and estimation nodes. 
Instead of streaming raw measurements to a distant TMC, each node exchanges local traffic‑state estimates only with its nearby V2X neighbors, thereby aligning the estimation architecture with the inherently local connectivity of V2X networks.
Second, we develop a distributed Kalman filtering framework tailored to the second‑order ARZ traffic flow model. 
This estimator fuses heterogeneous measurements from RSUs and CVs to reconstruct complex non‑equilibrium traffic phenomena, such as stop‑and‑go waves and nonlinear shockwave propagation. 
In doing so, it preserves the uncertainty‑aware, recursive structure of Kalman filtering while distributing computation and data fusion across the V2X network.
Third, we provide an extensive validation of the proposed framework under realistic V2X operating conditions. Using real-world trajectory datasets (HighD and NGSIM) together with SUMO simulations, we demonstrate that fixed infrastructure nodes can reconstruct dominant spatiotemporal traffic patterns and that an ego-CV can anticipate an approaching shockwave beyond its onboard sensing range by leveraging V2X-distributed state estimates. Finally, through a Monte Carlo study, we quantify how estimation accuracy depends jointly on CV penetration rate, V2X communication range, and realistic spatial infrastructure-sensor deployments. The results show that combined infrastructure and CV configurations improve upon estimation settings that rely exclusively on either fixed infrastructure or connected-vehicle, even under low CV penetration, finite V2X communication ranges, and sparse infrastructure coverage.

In our previous work \cite{Heij2025}, we presented a preliminary numerical study on a highway example to demonstrate the proof of concept of our DTSE algorithm. In contrast, the present work establishes theoretical guarantees and presents extensive numerical experiments on real data.


\subsection{Outline}
The remainder of the paper is organized as follows. Section~\ref{sec:traffic-model} presents the ARZ model and derives its discretized state-space formulation. Section~\ref{sec:sensor-network} describes the sensor network model for both RSUs and CVs. The proposed DTSE algorithm is presented in Section~\ref{sec:DTSE-algorithm}. Section~\ref{sec:experimental_setup} discusses the experimental setup. Section~\ref{sec:estimation_results} presents the estimation results under realistic V2X networks and traffic scenarios. Section~\ref{sec:statistical_evaluation} provides a statistical evaluation across different V2X networks, penetration rates, and RSU deployment configurations. Finally, Section~\ref{sec:conclusion} concludes the paper.

\subsection{Notation}
The set of real numbers is denoted by $\mathbb{R}$, and $\mathbb{R}^{n}$ and $\mathbb{R}^{m\times n}$ denote the spaces of real $n$-dimensional vectors and $m\times n$ matrices, respectively. For a matrix $A$, the transpose and inverse are written $A^{\top}$ and $A^{-1}$. A symmetric matrix $P = P^{\top} \in \mathbb{R}^{n\times n}$ is denoted as $P \succeq 0$ (resp., $P \succ 0$) if it is positive semi-definite (resp., positive definite); for two symmetric matrices $A, B$ we write $A \succeq B$ if
$A - B \succeq 0$. The $n\times n$ identity matrix is denoted by $I_{n}$, and
$e_{i}\in\mathbb{R}^{n}$ denotes the $i$-th canonical basis vector. The
Kronecker product is denoted by $\otimes$, and $\operatorname{diag}(A_{1},\dots,A_{m})$ denotes the block-diagonal matrix with blocks $A_{1},\dots,A_{m}$ on its diagonal. The Euclidean norm is denoted by $\|\cdot\|$, and $|\mathcal{S}|$ denotes the cardinality of a set $\mathcal{S}$. The expectation operator is $\mathbb{E}[\cdot]$, and $\mathcal{N}(\mu,\Sigma)$ denotes the Gaussian distribution with mean $\mu$ and covariance $\Sigma$. 

\section{Traffic Model} \label{sec:traffic-model}
This section reviews the ARZ model for describing highway traffic dynamics, which captures both the conservation of vehicles and the evolution of driver speed adaptation through a second-order macroscopic formulation. Then, following \cite{vishnoi2024}, we derive its discretized state-space formulation.


\subsection{ARZ Model} \label{subsec:ARZ}
The ARZ model \cite{aw2000,zhang2002} is a second-order macroscopic traffic flow model given by
\begin{align}
    \partial_t \rho + \partial_d (\rho v) &= 0  \label{eq:ARZpde1}\\
    \partial_t \psi +\partial_d \psi v &= -\frac{\rho(v-V_e(\rho))}{\tau} \label{eq:ARZpde2}
\end{align}
where $\rho(t,d)$ is the density (vehicles per unit road length) and $v(t, d)$ is the average speed (unit length per unit time), over time $t$ and spatial position $d$.  Let 
\begin{equation} \label{eq:driver-characteristic}
\chi \coloneqq v+p(\rho)
\end{equation}
denote the driver characteristic with $p(\rho)$ the pressure function that represents anticipatory driving behavior. Then, in~\eqref{eq:ARZpde2}, the auxiliary variable $\psi = \rho \chi$ denotes the relative flow, which is a momentum-like quantity that reflects deviations from the equilibrium traffic conditions, and  $\tau > 0$ is the relaxation time that determines how quickly drivers adapt their speed toward equilibrium. 

The pressure function is defined by
\begin{equation}\label{eq:pres_fun}
    p(\rho) = v_f \left(\rho / \rho_m \right)^\gamma,
\end{equation}
where $v_f$ is the free-flow speed, $\rho_m$ is the maximum jam density, and $\gamma>0$ is a fundamental diagram parameter.
Then, the equilibrium velocity at density $\rho$ is given by
\begin{equation}\label{eq:equilibrium-v-rho}
    V_e(\rho) = v_f \left(1-\left(\rho / \rho_m \right)^\gamma\right)
\end{equation}
The parameters $v_f, \rho_m, \gamma$ can be calibrated using empirical data and the fundamental diagram.

\subsection{Discretized State-Space Model}\label{subsec:ss}
Following \cite{vishnoi2024}, we discretize the ARZ model in space and time. The highway is partitioned into $N$ segments/cells of uniform length $\Delta h$, and the temporal domain is discretized using a sampling interval $\Delta t$. This results in a discretized state-space model with $\rho$ and $\psi$ as finite state vectors representing the average density and relative flow within the discretized cells. The discretized ARZ model can be written in state-space form as
\begin{align}\label{eq:ssARZ}
     x_{k+1} = Ax_{k} + Gf(x_k,u_k) + \omega_k
\end{align}
where $x_k = [x_{1,k}^\T,\ldots,x_{N,k}^\T]^\T\in \mathbb R^{2N}$ denotes the state vector of the highway discretized into $N$ cells, with 
$$
x_{i,k} = \begin{bmatrix}
 \rho_{i,k} &
 \psi_{i,k}
\end{bmatrix}^\T, \quad i=1,\ldots,N
$$
where $\rho_{i,k}$ is the average traffic density and $\psi_{i,k}$ the average relative flow within cell $i$ at time $k$. The variable $\omega_k$ denotes the process noise.

The input $u_k \in \mathbb{R}^3$ in \eqref{eq:ssARZ} specifies the boundary conditions of the model and is given by
\begin{equation}\label{eq:input}
   u_{k} = \begin{bmatrix}
   D_{0,k} & \chi_{0,k} & \rho_{{N+1},k} 
   \end{bmatrix}^\T
\end{equation}
where $D_{0,k}$ is the upstream demand, i.e., the number of vehicles per unit time entering the highway at time step $k$, $\chi_{0,k}$ is the driver characteristic \eqref{eq:driver-characteristic} at the upstream boundary, and $\rho_{N+1,k}$ specifies the density at the downstream boundary that governs the outflow of traffic. 

In standard TSE formulations, the input $u_k$ is typically assumed to be known at a centralized TMC. However, in the distributed setting considered here, only the boundary RSUs can directly measure the boundary conditions that determine $u_k$. 
In this paper, we assume that boundary RSUs broadcast $u_k$ to all nodes via a low-bandwidth, long-range communication channel (e.g., cellular LTE or LoRaWAN). 
This hierarchical communication architecture is well-established for ensuring scalability in heterogeneous vehicular networks \cite{abboud2016, zadobrischi2024}.
Since $u_k$ consists of only three scalar values, the bandwidth overhead for this global broadcast is negligible compared to the dense covariance matrices exchanged locally via short-range V2X.
In a practical deployment, this assumption can be relaxed by including the boundary variables in the consensus vector. Since the boundary RSUs can directly measure $u_k$, they act as leader nodes for these variables, allowing the true values to diffuse rapidly across the network via a consensus protocol \cite{ren2008}.

The system matrices $A,G \in \R^{2N\times 2N}$ in \eqref{eq:ssARZ} are given by
\[
A= I_N \otimes
\begin{bmatrix}
    1 & 0\\
    \frac{\Delta t\, v_f}{\tau} & 1-\frac{\Delta t}{\tau}
\end{bmatrix},
\quad G = I_N \otimes
\begin{bmatrix}
    \frac{\Delta t}{\Delta h} & 0\\
    0 & \frac{\Delta t}{\Delta h}
\end{bmatrix}
\]
where $\otimes$ denotes the Kronecker product, and 
\begin{equation*}
f(x_k,u_k)=
\biggl[
\begin{bmatrix}
q_{0,k}-q_{1,k}\\
\phi_{0,k}-\phi_{1,k}
\end{bmatrix}^\T,\ldots,
\begin{bmatrix}
q_{N-1,k}-q_{N,k}\\
\phi_{N-1,k}-\phi_{N,k}
\end{bmatrix}^\T
\biggr]^\T
\end{equation*}
represents net flows between adjacent cells. 
Here, $q_{i,k}$ denotes the traffic flux leaving cell $i$ and entering cell $i+1$ at time step $k$, measured in vehicles per unit time. Since we consider no on-ramps or off-ramps, the flux $q_{i,k}$ across the cell boundary is determined by the minimum of the upstream demand $D_{i,k}$ and the downstream supply $S_{i+1,k}$, i.e.,
$$
q_{i,k} = \min (D_{i,k}, S_{{i+1},k}).
$$
The variable $\phi_{i,k}$ denotes the relative flux given by
\[
\phi_{i,k} = q_{i,k} \chi_{i,k} = q_{i,k}\dfrac{\psi_{i,k}}{\rho_{i,k}}.
\]
The demand $D_{i,k}$ describes the maximum flow that can exit cell $i$ and is given by
\begin{equation*}
D_{i,k} =
\begin{cases}
    \rho_{i,k}(\chi_{i,k}-p(\rho_{i,k})), &\text{if} \ \rho_{i,k} \leq \sigma (\chi_{i,k}) \\
    \sigma(\chi_{i,k})(\chi_{i,k}-p(\sigma(\chi_{i,k}))), &\text{if} \ \rho_{i,k} > \sigma(\chi_{i,k}) 
\end{cases}    
\end{equation*}
where the critical density $\sigma(\chi_{i,k})$ is given by
\[
\sigma(\chi_{i,k}) = \rho_m\para{\frac{\chi_{i,k}}{v_f(1+\gamma)}}^{1 / \gamma}.
\]
The supply $S_{i,k}$ describes the maximum flow that can enter cell $i$ from upstream. Unlike demand, which depends only on the local state, the supply also accounts for the upstream driver characteristic $\chi_{i-1,k}$, as upstream traffic conditions constrain admissible inflow. The supply
\begin{equation*}
S_{i,k} \!=\!
\begin{cases}
    \sigma(\chi_{i,k})(\chi_{i-1,k}-p(\sigma(\chi_{i-1
,k}))), \! & \!\!\! \text{if } \rho_{i,k} \!\leq\! \sigma(\chi_{{i-1},k}) \\
    \rho_{i,k}(\chi_{{i-1},k}-p(\rho_{i,k})), \! & \!\!\! \text{if } \rho_{i,k} \!>\! \sigma(\chi_{{i-1},k}).
\end{cases}
\end{equation*}
As with demand, the critical density $\sigma$ separates the free-flow and congested regimes. When $\rho_{i,k}$ is below the upstream critical value $\sigma(\chi_{i-1,k})$, the downstream cell can accept additional inflow. When the density exceeds this threshold, congestion restricts the admissible inflow.

\section{Sensor Network Model} \label{sec:sensor-network}

Both fixed RSUs and CVs are assumed to measure the density $\rho_{i,k}$ and relative flow $\psi_{i,k}$ of the cell they occupy in the highway. Fixed RSUs always observe a predetermined segment, while CVs provide measurements only for the segment they currently occupy.
CVs measure their individual speed directly using GPS or odometry. 
They can also utilize on-board perception systems, such as camera-based vehicle counting and spacing estimation, to infer a local traffic density $\rho_{i,k}$ in cell~$i$. 
The average speed $v_{i,k}$ of cell~$i$ can be approximated by tracking and aggregating the individual speeds of all vehicles currently in that cell. 
The relative flow $\psi_{i,k}$ cannot be directly observed, but it can be reconstructed from the CVs' local measurements by substituting the measured $\rho_{i,k}$ and the approximated $v_{i,k}$ into the model: $\psi_{i,k} = \rho_{i,k} \chi_{i,k} = \rho_{i,k} (v_{i,k} + p(\rho_{i,k}))$, where $\chi_{i,k}$ is the driver characteristic in \eqref{eq:driver-characteristic} and $p(\rho_{i,k})$ is the anticipatory pressure term defined in \eqref{eq:pres_fun}.

\subsection{Measurement Equation}

The measurement equation for sensor $l\in \mathcal S$, where $ \mathcal{S}$ is the set of all sensors (both RSUs and CVs), is given by
\begin{equation}\label{eq:spatmeas}
y^l_k = C_k^l x_k + \nu^l_k
\end{equation}
where $y^l_k \in \mathbb{R}^2$ is the measurement vector collected by sensor $l$ at time $k$, $x_k \in \mathbb{R}^{2N}$ is the global traffic state vector at time $k$ evolving according to \eqref{eq:ssARZ}, $\nu^l_k \in \mathbb{R}^2$ is the measurement noise associated with sensor $l$. 
The dimension of $y^l_k$ is $\mathbb{R}^2$ because each sensor~$l$ measures the two state components, $\rho_{i_{l,k},k}$ and $\psi_{i_{l,k},k}$, associated with cell $i_{l,k}\in\{1,\dots, N\}$ where sensor~$l$ is located at time~$k$.

The measurement matrix $C_k^l \in \mathbb{R}^{2 \times 2N}$ is dynamic and sparse. 
It extracts the state vector block $x_{i_{l, k}, k}$ from $x_k$, where $i_{l,k}$ denotes the index of the specific cell currently occupied by sensor $l$ at time $k$. 
It is constructed as
$$C_k^l = e_{i_{l,k}}^\top \otimes I_2$$
where $i_{l,k} \in \{1, \dots, N\}$ is the cell index occupied by sensor $l$ at time $k$  and $e_{i_{l,k}}$ denotes the $i_{l,k}$-th canonical basis vector of $\mathbb{R}^N$.
This ensures that the measurement $y^l_k$ is directly mapped to the two-component state vector ($\rho_{i_{l,k},k}$ and $\psi_{i_{l,k},k}$) of the cell it occupies at time $k$. 
As RSUs are fixed, the corresponding $i_{l, k}=i_l$ is constant with respect to time $k$. 
For CVs, $i_{l, k}$ is updated based on the vehicle's position at time $k$. When a CV is outside the modelled segment, no such cell index is defined and we set $C_k^l = 0$.

\subsection{Communication Graph}
The communication network consists of persistent P2P links between neighboring RSUs and short-range V2X links involving nearby CVs. Specifically, each RSU is connected to its immediate RSU neighbors via a P2P network, which keeps the fixed infrastructure connected for the consensus-based information exchange even when no CVs are nearby. In addition, both CVs and RSUs establish short-range communication links with other CVs within their respective communication ranges. An illustrative example of this communication architecture is shown in Fig.~\ref{fig:scenario}.



For designing and analyzing the information flow in our DTSE algorithm, the sensor network (comprising both fixed RSUs and mobile CVs) is modeled as a dynamic undirected graph $\mathcal{G}_k = (\mathcal{S}_k, \mathcal{E}_k)$ at each time step $k$.
Here, $\mathcal{S}_k \subseteq \mathcal{S}$ is the {subset of sensor nodes active at time $k$} {(i.e., the sensors currently occupying a highway cell)}, and an edge $(l, m) \in \mathcal{E}_k$ exists if sensor nodes $l, m \in \mathcal{S}_k$ are within the physical V2X communication range of each other, establishing a bidirectional communication link for state estimate exchange.
The neighbor set of an active sensor $l \in \mathcal{S}_k$ is $\mathcal{N}_k^l = \{ m \in \mathcal{S}_k : (l,m) \in \mathcal{E}_k \}$, whereas an inactive CV $l \in \mathcal{S} \setminus \mathcal{S}_k$ does not form such edges and has $\mathcal{N}_k^l = \emptyset$.

Joint connectivity of $\mathcal{G}_k$, which is required for the convergence analysis of the proposed algorithm in Section~\ref{sec:DTSE-algorithm}, is defined in the next section.
Moreover, we adopt the simplifying technical assumption that the volume of transmitted data (the local state estimates and covariance matrices) does not exceed the capacity of any link \cite{olfati-saber2007}. 
In this paper, this assumption is justified by the relatively low instantaneous CV penetration rate expected in the considered operational scenario. 
This keeps the node density and, consequently, the number of simultaneous communication sessions low, preventing the channel congestion and packet loss typically observed in dense vehicular ad-hoc networks \cite{kenney2011, abboud2016}. 
This allows us to neglect link-saturation effects and focus solely on the connectivity topology necessary for the distributed estimation consensus.

\subsection{Observability Assumptions}
Observability of the linearized ARZ state-space model in \eqref{eq:ssARZ} from the available RSU and CV measurements is required for convergence of the algorithm in Section~\ref{sec:DTSE-algorithm}.
A time-varying observability Gramian analysis~\cite{vishnoi2024} shows that full observability requires the most downstream mainline cell to be sensed at all times, as is typical in conservation-law-based traffic models. We therefore refer to this requirement as collective observability of the linearized system. In the simulations, we assume that a fixed RSU is always present on the most downstream mainline segment, which is consistent with the boundary-RSU assumption already introduced in Section~\ref{subsec:ARZ}.
In the distributed setting, observability need not hold locally at every node because upstream sensors do not directly access downstream RSU measurements. The P2P V2X communication structure helps propagate this information across the network over successive time steps. However, when CV penetration is low or communication ranges are short, the network may be too sparse for downstream information to reach all upstream nodes reliably. Section~\ref{sec:statistical_evaluation} examines the practical impact of distributed observability under varying RSU density, CV penetration, and V2X communication range.

\section{Distributed Traffic State Estimation Algorithm and Guarantees} 
\label{sec:DTSE-algorithm}
Distributed state estimation over sensor networks is a well-established problem in control theory. Foundational work on distributed Kalman filtering introduced consensus-based schemes in which sensor nodes estimate system states by exchanging measurements and uncertainty information~\cite{olfati-saber2005, olfati-saber2007}. Subsequent studies addressed stability under communication constraints~\cite{carli2008}, diffusion-based information propagation~\cite{cattivelli2010}, convergence under weak observability~\cite{das2016}, event-triggered communication~\cite{battistelli2018}, integrated estimation and control~\cite{talebi2019}, and robustness to model uncertainty~\cite{zorzi2019}.

Most of these methods, however, are formulated for linear systems. In traffic-flow applications, the estimator must also handle nonlinear dynamics and physical constraints such as non-negativity and saturation. The algorithm developed below builds on the information-form DKF~\cite{battistelli2018} and adapts it to the nonlinear ARZ model through linearization and constrained state estimation.



\subsection{Algorithm}
Let the process noise be given by $\omega_k \sim \mathcal{N}(0,Q)$, where $Q \in \mathbb{R}^{2N \times 2N}$ is the positive definite process noise covariance matrix. Similarly, for each sensor node $l \in \mathcal{S}$, let the local measurement noise be given by $\nu_k^l \sim \mathcal{N}(0,R_k^l)$, where $R_k^l \in \mathbb{R}^{2 \times 2}$ is positive definite. Furthermore, let $\bar{x}_0 = \mathbb{E}[x_0]$ denote the prior expectation of the initial traffic state, and let $P_0 \in \mathbb{R}^{2N \times 2N}$, with $P_0 \succ 0$, denote the corresponding initial error covariance. At time $k=0$, each sensor node $l \in \mathcal{S}$ is initialized with the prior information pair
\begin{align*}
    \Xi_{0|-1}^l = (P_0)^{-1},\quad
    \hat{x}_{0|-1}^l = \Pi_{\mathcal D}(\bar{x}_0),\quad
    \xi_{0|-1}^l = \Xi_{0|-1}^l \hat{x}_{0|-1}^l,
\end{align*}
where $\Pi_{\mathcal{D}}$ is the projection operator defined later in \eqref{eq:projection}

Then, for each time $k \geq 0$ and for each sensor node $l \in \mathcal{S}$, we perform the following five steps:

\subsubsection*{Step 1 - Linearization}
We linearize \eqref{eq:ssARZ} around the current local prior estimate $\hat{x}_{k|k-1}^l = (\Xi_{k|k-1}^l)^{-1} \xi_{k|k-1}^l$ as
\begin{subequations}
\label{eq:linearized_dyn}
\begin{align}
\Lambda^l_k &= A + G \left. \frac{\partial f(x, u_k)}{\partial x} \right|_{x = \hat{x}_{k|k-1}^l}
\label{eq:linearized_dyn-Lambda} \\
\eta^l_k &= G \left( f(\hat{x}_{k|k-1}^l, u_k) - \left[ \left. \frac{\partial f(x, u_k)}{\partial x} \right|_{x = \hat{x}_{k|k-1}^l} \right] \hat{x}_{k|k-1}^l \right)
\label{eq:linearized_dyn-eta}
\end{align}
\end{subequations}
where $\Lambda^l_k$ is the linearized state transition matrix and $\eta^l_k$ is the offset vector.

\subsubsection*{Step 2 - Local Measurement Update}
Each node assimilates its own sensor data into the information space.
The local information contribution $\theta_k^l$ and associated precision matrix $\Theta_k^l$ are computed as
$$
\theta_k^l = (C_k^l)^\T (R_k^l)^{-1} y_k^l, \quad \Theta_k^l = (C_k^l)^\T (R_k^l)^{-1} C_k^l .
$$
Then, the local posterior updates are given by
\begin{equation}
    \label{eq:local-update}
    \xi_{k|k}^l = \xi_{k|k-1}^l + \theta_k^l, \quad \Xi_{k|k}^l = \Xi_{k|k-1}^l + \Theta_k^l.
\end{equation}

\subsubsection*{Step 3 - Information Fusion}
Nodes exchange their local information pairs $(\xi_{k|k}^l, \Xi_{k|k}^l)$ with neighbors $j \in \mathcal{N}^l_k$.
An iterative consensus step is performed for $L$ iterations, where $L$ is chosen according to the available communication budget and the desired number of information-exchange hops per time step.
Let $\alpha = 0, \dots, L$ be the consensus iteration index. Initialize $\xi^{l,(0)}_{k|k} = \xi^l_{k|k}$ and $\Xi^{l,(0)}_{k|k} = \Xi^l_{k|k}$.
Then, at iteration $\alpha=1, \dots, L$, compute
\begin{subequations}
\label{eq:multi-hop-consensus}
    \begin{align}
        \xi^{l,(\alpha)}_{k|k} &= \sum_{j \in \mathcal{N}^l_k \cup \{l\}} \pi_{(l,j),k}  \xi^{j,(\alpha-1)}_{k|k} \\
        \Xi^{l,(\alpha)}_{k|k} &= \sum_{j \in \mathcal{N}^l_k \cup \{l\}} \pi_{(l,j),k}  \Xi^{j,(\alpha-1)}_{k|k}
    \end{align}
\end{subequations}
where $\pi_{(l,j),k}$ are chosen as Metropolis weights. Let $d^l_k = |\mathcal{N}^l_k|$ be the degree of node $l$ in the graph $\mathcal{G}_k$. Then
\[
\pi_{(l,j),k} =
\begin{cases}
\frac{1}{1+\max\{d^l_k,d^j_k\}}, & j \in \mathcal{N}^l_k \\
1 -\sum_{i\in\mathcal{N}^l_k}\pi_{(l,i),k}, & j = l \\
0, & \text{otherwise.}
\end{cases}
\]
Since $\mathcal{G}_k$ is undirected, the resulting consensus matrix is doubly stochastic.

After $L$ iterations, the fused information variables are given by
\begin{equation}
\label{eq:fused}
\overline{\xi}^l_{k|k} = \xi^{l,(L)}_{k|k}, \quad \overline{\Xi}^l_{k|k} = \Xi^{l,(L)}_{k|k}.
\end{equation}

\subsubsection*{Step 4 - Prediction}
Using the linearized dynamics \eqref{eq:linearized_dyn} and the fused information \eqref{eq:fused}, we compute the prior for the next time step. We employ the matrix inversion lemma to propagate the information matrix without direct inversion of the state covariance, i.e.,
\begin{equation}
    \label{eq:pred-cov}
    \Xi^l_{k+1|k} = Q^{-1} - Q^{-1} \Lambda^l_k M_k^l (\Lambda^l_k)^\T Q^{-1}
\end{equation}
where $M_k^l = \left( \overline{\Xi}^l_{k|k} + (\Lambda^l_k)^\T Q^{-1} \Lambda^l_k \right)^{-1}$.
The predicted information vector accounts for $\eta^l_k$ in \eqref{eq:linearized_dyn-eta} as
\begin{equation}
    \label{eq:pred-info-vec}
    \xi^l_{k+1|k} = \Xi^l_{k+1|k} \left( \Lambda^l_k (\overline{\Xi}^l_{k|k})^{-1} \overline{\xi}^l_{k|k} + \eta^l_k \right).
\end{equation}
 
\subsubsection*{Step 5 - Enforcing Physical Constraints}
To ensure the traffic state estimate remains physically meaningful, we project the predicted state $\hat{x}^l_{k+1|k} = (\Xi^l_{k+1|k})^{-1} \xi^l_{k+1|k}$ onto the feasible state space \cite{vishnoi2024}.
For each cell $i \in \{1, \dots, N\}$, $\hat{x}^l_{i, k+1|k} = [\hat{\rho}^l_{i, k+1|k} ~ \hat{\psi}^l_{i, k+1|k}]^\T$ is projected as
\begin{subequations}
\label{eq:projection}
\begin{align}
\hat{\rho}^l_{i, k+1|k} &\leftarrow \min \Big( \max(\hat{\rho}^l_{i, k+1|k}, 0), \rho_m \Big) \\
\hat{\psi}^l_{i, k+1|k} &\leftarrow \min \Big( \max(\hat{\psi}^l_{i, k+1|k}, 0), v_f \rho_m \Big).
\end{align}
\end{subequations}
Thus, the predicted estimate is updated as
$$\hat{x}^l_{k+1|k} \leftarrow \Pi_{\mathcal{D}}((\Xi^l_{k+1|k})^{-1} \xi^l_{k+1|k})$$
where $\Pi_\mathcal{D}$ is the projection operator defined in \eqref{eq:projection}.
The information matrix $\Xi_{k+1|k}^l$ remains unchanged to preserve uncertainty characteristics, but the information vector is updated to reflect the constrained state as
$\xi^l_{k+1|k} \leftarrow \Xi^l_{k+1|k} \, \hat{x}^l_{k+1|k}.$

\subsection{Stochastic Stability and Covariance Bound}
The projection operator in \eqref{eq:projection} ensures that the state estimates remain within the physical domain $\mathcal{D} = [0, \rho_m] \times [0, v_f \rho_m]$. However, this property alone only guarantees that the estimation error cannot exceed the diameter of $\mathcal{D}$. 
A non-trivial, mathematically rigorous guarantee of convergence must be driven by the sensor network topology and measurement quality rather than mere state-space truncation. To this end, we analyze the covariance contraction induced by the distributed information fusion.

To formalize the stability of the proposed DKF for the ARZ model, we first introduce two fundamental definitions.

\begin{define}
    \label{def:joint-connectivity}
    The sensor network $\mathcal{G}_k = (\mathcal{S}_k, \mathcal{E}_k)$ is said to be \emph{jointly connected} over uniformly bounded time intervals if there exists an integer $b > 0$ such that for all $k \geq 0$, the union graph $\bigcup_{t=k}^{k+b-1} \mathcal{G}_t$ is connected.
\end{define}

Joint connectivity implies that over any time window of length $b$, information from any node can propagate to every other node in the network via the multi-hop consensus protocol \eqref{eq:multi-hop-consensus}.
Let 
\begin{subequations}
\label{eq:C_k-R_k}
    \begin{align}
        C_k &= \begin{bmatrix} (C_k^1)^\top, \ldots, (C_k^{|\mathcal{S}|})^\top\end{bmatrix}^\top  \in \mathbb{R}^{2|\mathcal{S}|\times 2N}
        \label{eq:C_k}
        \\
        R_k &= \text{diag}(R_k^1, \ldots, R_k^{|\mathcal{S}|}) \in\mathbb{R}^{2|\mathcal{S}|\times 2|\mathcal{S}|}
        \label{eq:R_k}
    \end{align}
\end{subequations}
denote the concatenated measurement matrix and block diagonal noise covariance matrix, respectively. 
Let $\Lambda_k$ be the linearized dynamics around the true state trajectory $x_k$, i.e.,
\begin{equation}
    \label{eq:Lambda_k}
    \Lambda_k = A + G \left. \frac{\partial f(x, u_k)}{\partial x} \right|_{x = x_k} \in\mathbb{R}^{2N\times 2N}.
\end{equation}

\begin{define}
    \label{def:obs}
    The global linearized system $(\Lambda_k,C_k)$ is \emph{collectively observable} if there exists an integer $N_o > 0$ and a constant $\underline{o} > 0$ such that the observability Gramian satisfies
    $$ 
    \mathcal{O}_{k, N_o} = \sum_{t=k-N_o}^{k} (\Phi_{t, k-N_o})^\top C_t^\top R_t^{-1} C_t \Phi_{t, k-N_o} \succeq \underline{o} I_{2N}
    $$
    for all $k \geq N_o$, where $\Phi_{t, k} = \Lambda_{t-1} \Lambda_{t-2} \dots \Lambda_k$ is the state transition matrix for $t > k$, and $\Phi_{k, k} = I_{2N}$.
\end{define}

In the context of the ARZ model \eqref{eq:ssARZ}, collective observability is structurally satisfied provided the most downstream mainline cell is consistently measured by an RSU. 


\begin{assume}
    \label{assume:bounded-nonlin-residual}
    The non-linear residual 
    \begin{multline*}
    \varphi(x_k, \hat{x}_k^l) = G \bigg[ f(x_k, u_k) - f(\hat{x}_k^l, u_k) \\ - \left. \frac{\partial f}{\partial x} \right|_{\hat{x}_k^l} (x_k - \hat{x}_k^l) \bigg]
    \end{multline*}
    of the ARZ dynamics \eqref{eq:ssARZ} is bounded such that, for some $\varrho > 0$,
    \begin{equation}
        \label{eq:lin-remainder-bound}
        \|\varphi(x_k, \hat{x}_k^l)\| \leq \varrho \|x_k - \hat{x}_k^l\|^2.
    \end{equation}
\end{assume}

This assumption follows directly from Taylor's theorem. The residual $\varphi$ constitutes the remainder of the first-order Taylor expansion of the non-linear traffic dynamics around the estimate $\hat{x}_k^l$. Because the flux and pressure functions of the ARZ model are twice continuously differentiable over the physical domain $\mathcal{D}$, their second derivatives (Hessian matrices) are uniformly bounded. Therefore, \eqref{eq:lin-remainder-bound} holds, where the constant $\varrho > 0$ is proportional to the supremum of the Hessian norm over $\mathcal{D}$.

For the theoretical stability analysis, we evaluate the network over a persistently active subset of sensor nodes, $\mathcal{S}^* \subseteq \mathcal{S}$ (e.g., the RSUs and a fleet of CVs currently traversing the segment). Definitions~\ref{def:joint-connectivity} and \ref{def:obs}, as well as the bounds in the theorem below, apply to this fixed-dimension sub-network.

    
\begin{thm}
\label{thm:error-guarantee}
Suppose the global linearized system $(\Lambda_k, C_k)$ is collectively observable, the communication graph $\mathcal{G}_k$ is jointly connected over uniformly bounded intervals, and Assumption~\ref{assume:bounded-nonlin-residual} holds.
Then, the fused information matrices are uniformly bounded, satisfying $0 < \underline{\sigma} I_{2N} \preceq \overline{\Xi}_{k|k}^l \preceq \overline{\sigma} I_{2N}$ for all $l \in \mathcal{S}$ and $k \geq 0$. Furthermore, provided the initial estimation errors and noise covariances are sufficiently small, there exist constants $c_1 > 0$, $\lambda \in (0, 1)$, and $M > 0$ such that the expected mean-square error for any node $l \in \mathcal{S}$ satisfies
\begin{equation}
    \label{eq:error-guarantee}
    \mathbb{E}[\|x_k - \hat{x}_k^l\|^2] \leq c_1 \left( \sum_{j \in \mathcal{S}} \|x_0 - \hat{x}_0^j\|^2 \right) \lambda^k + M, \quad \forall k \geq 0
\end{equation}
where the asymptotic bound $M$ only depends on noise and is independent of the maximum density $\rho_m$.
\end{thm}
\begin{proof}
    See Appendix~\ref{appendix:proof-error-guarantee}.
\end{proof}

Theorem~\ref{thm:error-guarantee} shows that the proposed DTSE algorithm is stochastically bounded. Its mean-square error has two parts: an exponentially decaying transient with rate $\lambda$ and a steady-state noise floor $M$. Hence, the estimator tracks the macroscopic traffic state without diverging and ultimately remains in a bounded uncertainty region determined by sensing quality and network connectivity. Reducing this bound requires either lower sensor variance, which decreases $c'$, or stronger V2X connectivity, which improves $\lambda$ and $\underline{\sigma}$.

The small-error and bounded-noise conditions are standard local assumptions for extended Kalman filtering. Because the ARZ model is nonlinear, DTSE uses a local Taylor expansion; the initial-error condition defines a region of attraction in which the linear contraction dominates the nonlinear remainder. The bounded-noise condition ensures that routine disturbances from driver behavior or sensor imperfections do not push the estimate out of this neighborhood. Within this region, the estimator remains stable and converges over time.


\section{Experimental Setup}\label{sec:experimental_setup}
The proposed DTSE method is evaluated on three highway scenarios. The HighD and NGSIM datasets provide empirical trajectory data under realistic highway traffic conditions and are used to evaluate reconstruction of density and velocity fields. The SUMO scenario reproduces the motivating case illustrated in Fig.~\ref{fig:scenario}: a downstream bottleneck generates a backward-propagating shockwave while the ego-CV enters from upstream. This larger and controllable highway domain allows evaluation of ego-CV estimation beyond the vehicle's onboard sensing range.

Together, these scenarios support both qualitative and statistical evaluation. First, representative node-level estimates are analyzed to assess whether the DTSE framework reconstructs the dominant spatiotemporal traffic dynamics at fixed RSU and mobile ego-CV nodes. Second, network-level performance is evaluated by varying CV penetration rate, V2X communication range, and RSU deployment density, thereby quantifying how sensing availability and communication connectivity affect estimation accuracy.

\subsection{Scenarios and Ground-Truth Data}\label{subsec:data}
Here we provide details on the two empirical datasets and
the simulated SUMO scenario used in our experiments.

\begin{itemize}
    \item HighD. Track 25 from Location 1 records a highway near Cologne with three lanes in each direction. The three lanes carrying right-to-left traffic are taken over the full 400~m segment, from 8:55-9:07~am.
    \item NGSIM (US-101). The recording covers a 550~m, five-lane segment, including one auxiliary lane for on- and off-ramps whose flows were negligible during the selected period. The data span 8:07-8:19~am and are restricted to the 100-500~m range, giving a 400~m segment comparable in length to HighD.
    \item SUMO. A two-lane, $2.7$~km highway in which the first and last $100$~m serve as buffer zones for extracting upstream demand and downstream density boundary conditions\footnote{For the empirical scenarios, the first and last cells serve as boundary cells to maximize the analyzed highway length.}, yielding an effective study domain of $2.5$~km. The simulation runs for $1200$~s with a speed limit of $100$~km/h. Traffic demand at the upstream boundary is $3600$~veh/h ($1$~s average interarrival time), with vehicles inserted via SUMO's best-lane departure rule at maximum speed. Vehicle dynamics follow the standard Krauss car-following model with a $1.5$~m minimum gap. To trigger transient, non-equilibrium conditions, a temporary bottleneck is induced at $d = 2200$~m by lowering the speed limit from $100$~km/h to $10$~km/h during $t \in [700,760]$~s, producing a stop-and-go shock wave that propagates upstream.
\end{itemize}

For all three scenarios, the raw trajectory data are aggregated into an Eulerian traffic state using Edie's generalized definition on space-time cells. The HighD and NGSIM segments are discretized with $t_e = 2$~s and $\Delta h = 40$~m, yielding $N = 10$ cells over a $720$~s observation window, while the SUMO scenario uses $t_e = 5$~s and $\Delta h = 100$~m, yielding $N = 25$ cells over the $1200$~s horizon. For each cell $(i,k)$, the total travel time $T_{i,k}$ and traveled distance $H_{i,k}$ of all assigned trajectory segments are accumulated, and the density and speed are computed as $\rho_{i,k} = T_{i,k}/(\Delta h\,t_e)$ and $v_{i,k} = H_{i,k}/T_{i,k}$. The auxiliary relative-flow state is then $\psi_{i,k} = \rho_{i,k}\chi_{i,k}$, with $\chi_{i,k}$ as in \eqref{eq:driver-characteristic} and $p(\rho)$ as in \eqref{eq:pres_fun}; while $\rho_{i,k}$ and $v_{i,k}$ follow directly from the Edie aggregation, $\psi_{i,k}$ additionally requires the calibrated FD parameters $v_f$, $\rho_m$, and $\gamma$. 

The ground-truth state for each scenario is shown in the left column of Figs.~\ref{fig:highd_rsu1}, \ref{fig:ngsim_rsu1}, and \ref{fig:SUMO_egoCV}. The HighD and NGSIM both capture congested-flow conditions with multiple backward-propagating stop-and-go waves, with NGSIM exhibiting notably higher densities (up to $500$~veh/km versus $260$~veh/km in HighD) and more pronounced congestion. The SUMO
scenario, by contrast, evolves from largely free-flow conditions to a sharp, accident-induced bottleneck that triggers a single backward shockwave.

\subsection{Parameter Calibration}\label{subsec:calibration}

The ARZ model parameters are calibrated against the ground-truth traffic state following~\cite{seo2017}. The FD parameters are obtained by nonlinear least-squares regression of the equilibrium relation in \eqref{eq:equilibrium-v-rho} onto nearly-stationary traffic states, defined as the mean speed and density of each cell whose within-cell speed coefficient of variation is below a threshold. This threshold is $25\%$ for HighD and NGSIM, and $50\%$ for SUMO to retain enough samples from its accident-induced congested regime.


The resulting estimates for the HighD scenario are $v_f = 50.90$ km/h, $\rho_m = 237.15$ veh/km, and $\gamma = 1.8564$. For the NGSIM scenario, the estimates are $v_f = 54.03$ km/h, $\rho_m = 416.14$ veh/km, and $\gamma = 2.2586$. For the SUMO scenario, the estimates are $v_f = 95.31$ km/h, $\rho_m = 232.56$ veh/km, and $\gamma = 1.1882$. The relaxation time is set to $\tau = 40$ s for NGSIM following the calibration in~\cite{belletti2015}, and likewise to $\tau = 40$ s for HighD. For SUMO we use $\tau = 20$ s. Numerical stability of the discretization is ensured by verifying the Courant-Friedrichs-Lewy (CFL) condition~\cite{courant1967}, $\frac{v_f\,\Delta t}{\Delta h} < 1,$ which is satisfied for all three scenarios. HighD and NGSIM use $\Delta t = 1$ s and $\Delta h = 40$ m, while SUMO uses $\Delta t = 1$ s and $\Delta h = 100$ m.


\subsection{Sensing and V2X Network Configuration}\label{subsec:sensors}

The sensing and communication setup specifies the placement of the RSUs, the CV penetration rate, and the V2X and P2P links that connect them. As discussed in Section~\ref{sec:sensor-network}, RSUs maintain persistent P2P links with their immediate neighbors, independent of the V2X channel, while all RSU-CV and CV-CV communication is governed by a finite V2X communication range. Empirical scenarios use a denser infrastructure over a short highway segment, whereas the SUMO scenario uses a sparser deployment over a much longer domain to accommodate the ego-CV setting.

For the HighD and NGSIM scenarios, the fixed infrastructure consists of two RSUs at $d \in \{20, 380\}$~m with a V2X communication range of $300$~m, consistent with ranges reported in real-world highway V2X evaluations~\cite{maglogiannis2022}. The CV penetration rate is set to approximately $10\%$, yielding $77$ of $774$ vehicles ($9.95\%$) for HighD and $167$ of $1668$ ($10.01\%$) for NGSIM. In both scenarios, the upstream RSU at $20$~m serves as the representative estimator for the analysis.

The SUMO scenario covers a substantially larger domain and therefore uses four RSUs positioned at $d \in \{50, 850, 1650, 2450\}$~m, giving an inter-RSU spacing of $800$~m, larger than in the empirical setups. The V2X communication range is set to $400$~m, deliberately shorter than the inter-RSU spacing, so any state information an RSU contributes to its peers can only reach intermediate vehicles through CV-mediated relays. This design stresses the consensus protocol's ability to fuse information across the mobile fleet. The CV penetration rate is again set to approximately $10\%$.
The analysis focuses on the time window $t \in [700, 845]$~s, defined by the traversal of a reference ego-CV that enters the network as the congestion wave is triggered. During this interval, $248$ vehicles traverse the highway segment, of which $25$ are designated as CVs ($10.08\%$ penetration).


\subsection{DTSE Filter Tuning and Initialization}\label{subsec:filter}
The DTSE filter is initialized uniformly across all sensor nodes $l \in \mathcal{S}$. At $k=0$, the prior mean is chosen as $\bar{x}_0 = \mathbb{E}[x_0] = x_0$, with initial error covariance $P_0 = 10^{-3} I_{2N}$. This provides a common reference across configurations, so the reported results reflect estimator performance in a post-transient regime rather than an arbitrary initial mismatch, whose influence is confined to a short start-up phase.
The measurement noise covariance is defined, for all $l \in \mathcal{S}$, as $R^{l}=\operatorname{diag}\!\left((\beta\rho_{\max})^{2},(\beta\psi_{\max})^{2}\right)$, with $\beta=0.01$. To account for ARZ model mismatch and discretization errors, the process noise covariance is set to $Q=I_N\otimes \operatorname{diag}\!\left((\kappa\rho_{\max})^{2},(\kappa\psi_{\max})^{2}\right)$, with $\kappa=0.01$. To ensure physical realizability, the state estimates are projected onto the box constraints $\rho\in[0,\rho_m]$ veh/km and $\psi\in[0,v_f\rho_m]$ veh/h. Finally, the consensus protocol uses $L=5$ communication rounds per time step to accelerate information fusion across the sensor nodes.


\section{Estimation Results}\label{sec:estimation_results}

We now assess the qualitative reconstruction quality of the proposed DTSE algorithm under the baseline configuration described in
Section~\ref{sec:experimental_setup}. For each scenario, we examine the spatiotemporal density and speed fields maintained by the
representative upstream RSU and ego-CV node.

\subsection{RSU State Estimation}
\label{subsec:results_rsu}

Figures~\ref{fig:highd_rsu1} and~\ref{fig:ngsim_rsu1} compare the ground-truth traffic fields with the DTSE reconstructions for the HighD and NGSIM datasets, respectively. In each figure, the left column shows the ground-truth density $\rho(t,d)$ and speed $v(t,d)$, derived from Edie's generalizations, and the right column shows the corresponding RSU estimates $\hat{\rho}(t,d)$ and $\hat{v}(t,d)$. 

Despite the sparse $10\%$ CV penetration rate, the DTSE method recovers the dominant wave locations, propagation slopes, and recurrence patterns in both datasets. Because HighD carries less traffic over the same window, its $77$ CVs (of $774$ vehicles) provide substantially fewer absolute measurements than NGSIM's $167$ (of $1668$). The estimated fields still preserve the alternating high-density/low-speed and low-density/high-speed regions, showing that the RSU reconstructs traffic  dynamics beyond its own sensing location through the distributed consensus network.

Across both datasets, the reconstructed fields are smoother than the ground truth and omit some fine-scale variations, as expected from the model-based estimator and the projection onto the admissible state space,
which attenuates local densities outside the feasible model range. 

\begin{figure}[!t]
    \centering
    \includegraphics[width=\linewidth]{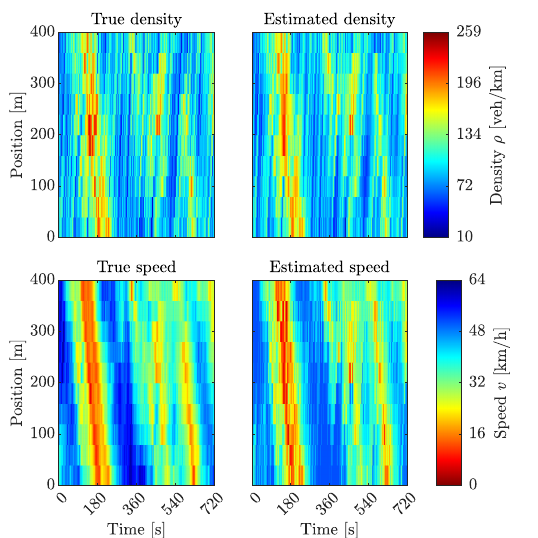}
    \caption{Spatiotemporal reconstruction of highway density and speed by the first RSU. The left column shows the HighD density and speed fields, and the right column shows the corresponding RSU estimate.}
    \label{fig:highd_rsu1}
\end{figure}

\begin{figure}[!t]
    \centering
    \includegraphics[width=\linewidth]{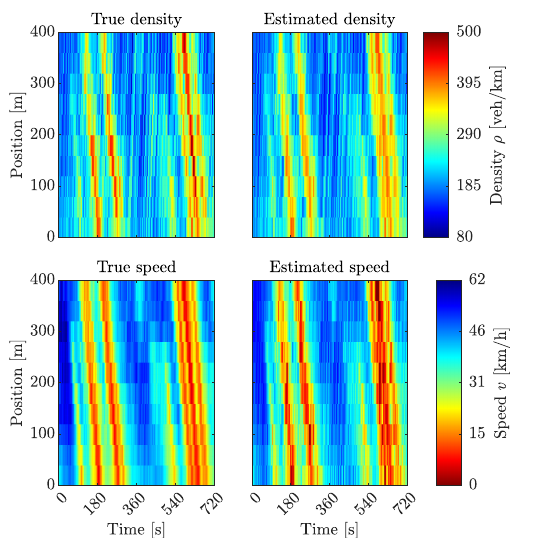}
    \caption{Spatiotemporal reconstruction of highway density and speed by the first RSU. The left column shows the NGSIM density and speed field, and the right column shows the RSU estimate.}
    \label{fig:ngsim_rsu1}
\end{figure}

\begin{figure}[!t]
    \centering
    \includegraphics[width=\linewidth]{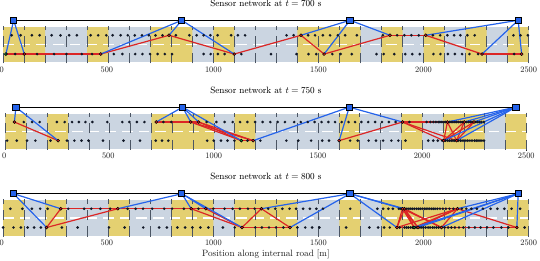}
\caption{Vehicle-sensor communication network at $t = 700$, $750$, and $800$~s: RSUs (blue squares), CVs (red dots), conventional vehicles (black dots), the ego-CV (green dot), and local sensing ranges (yellow). Links
denote P2P (black), V2V (red), and V2I (blue) communication. The ego-CV's accessible information changes as traffic evolves and links form and break.}
    \label{fig:SUMO_network}
\end{figure}

\begin{figure}[!t]
    \centering
    \includegraphics[width=\linewidth]{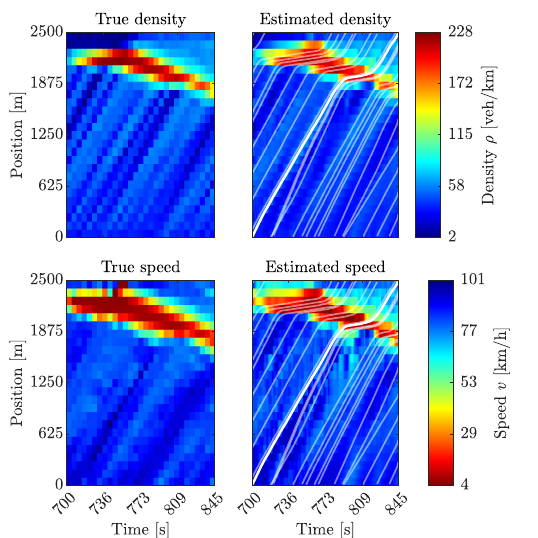}
    \caption{Spatiotemporal reconstruction of highway density and speed by the ego-CV. The left column shows the SUMO density and speed field, and the right column shows the ego-CV estimate. The bold white line is the trajectory of the ego-CV and the other white lines are the trajectories of the CV fleet.} 
    \label{fig:SUMO_egoCV}
\end{figure}

\subsection{Ego-CV State Estimation}\label{subsec:results_egocv}
We evaluate the SUMO scenario by examining both the evolving sensing and communication graph the ego-CV operates within and the resulting spatiotemporal reconstruction it maintains onboard.

Figure~\ref{fig:SUMO_network} provides a sequence of snapshots illustrating the dynamic evolution of the sensing and communication graph $\mathcal{G}_k$. The distinct node types are color-coded: fixed RSUs (blue squares), participating CVs (red dots), and non-connected vehicles (black dots). The ego-CV, whose internal estimation belief is the subject of this analysis, is highlighted in green. The yellow shaded regions delineate the instantaneous observable subspace of the network, i.e., areas where density and flow are directly measurable by at least one sensor node. As evidenced by the gaps between yellow regions, the network is characterized by sparse and fragmented observability, and accurate estimation therefore relies heavily on the distributed consensus mechanism to diffuse information from observed segments to unobserved ones via V2X communication links.

A comparative analysis of the ground-truth traffic dynamics and the distributed reconstruction is presented in Fig.~\ref{fig:SUMO_egoCV}. The left panel visualizes the macroscopic density field $\rho(t,d)$ and speed field $v(t,d)$ serving as the ground truth, while the right panel displays the estimated fields $\hat{\rho}(t,d)$ and $\hat{v}(t,d)$ as reconstructed by the ego-CV, with the ego-CV trajectory overlaid for context. The imposition of the speed limit drop at $t = 700$~s triggers a rapid transition from free-flow to congested conditions, producing a backward-propagating shockwave (the high-density region slanting upwards to the left). Despite the ego-CV being located far upstream of the bottleneck when the congestion initiates, its local estimate captures the onset of the congestion wave almost simultaneously with the ground truth, indicating effective information dissemination from the RSUs and CVs through the multi-hop consensus network. The estimator further reproduces the propagation speed of the shockwave front and the subsequent dissipation of the jam after the speed limit is restored at $t = 760$~s. The diffusion-based consensus mechanism and the spatial discretization introduce slight smoothing relative to the ground truth, but the estimated fields preserve the main spatiotemporal extent and propagation of the high-density region. This confirms that the proposed nonlinear information filter enables mobile CV nodes to maintain traffic-state awareness beyond their onboard sensing range, even when the network topology is dynamic and direct observation coverage is partial.


   \begin{figure}[!t]
    \centering
    \includegraphics[width=\linewidth]{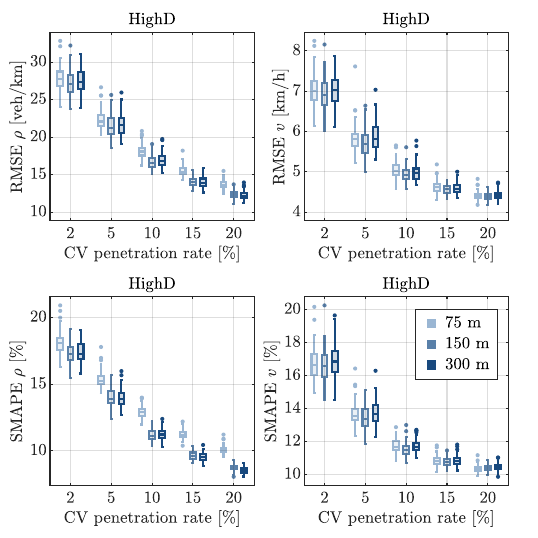}
    \caption{Estimation accuracy of the fused network estimate in SMAPE and RMSE for the HighD scenario across CV penetration rates and V2X communication ranges, over 100 trials.}
    \label{fig:highd_boxplot}
\end{figure}

\begin{figure}[!t]
    \centering
    \includegraphics[width=\linewidth]{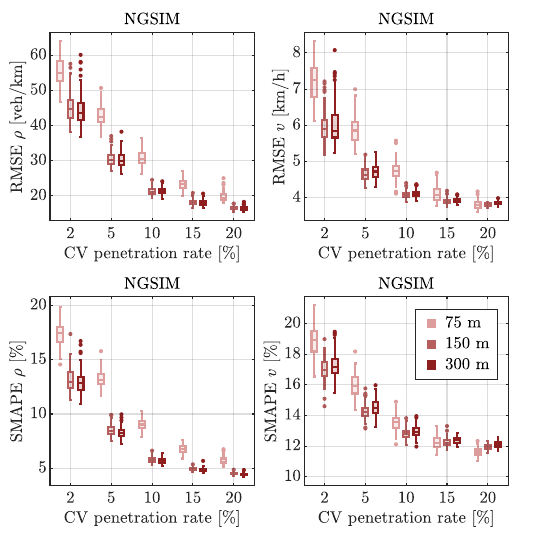}
    \caption{Estimation accuracy of the fused network estimate in SMAPE and RMSE for the NGSIM scenario across CV penetration rates and V2X communication ranges, over 100 trials.}
    \label{fig:ngsim_boxplot}
\end{figure}

\section{Statistical Evaluation}\label{sec:statistical_evaluation}

To complement the previous case studies, we evaluate how CV penetration rate, RSU deployment density, and V2X communication range affect estimation accuracy. Since the communication graph $\mathcal{G}_k$ depends on the random spatial distribution of participating CVs, we use Monte Carlo trials to characterize performance across different CV selections. 

We conduct two experiments. The first varies the CV penetration rate and the V2X communication range while keeping the RSU deployment fixed at the baseline configuration. The tested penetration rates are
$p \in \{2,5,10,15,20\}\%$. The tested ranges are
$r \in \{75,150,300\}$~m for the $400$~m HighD and NGSIM highway segments, and $r \in \{300,400,500\}$~m for the $2.5$~km SUMO scenario.

The second experiment varies RSU deployment density $d_0,\dots,d_5$ and CV penetration rate $p \in \{0,2,5,10,15,20\}\%$ at the baseline V2X communication
range ($300$~m for HighD and NGSIM, and $400$~m for SUMO). The RSU deployment configurations $d_0,\dots,d_5$ are listed in Table~\ref{tab:rsu_deployment}. In Fig.~\ref{fig:RSUvsCV}, the row $d_0$ corresponds to CV-only configurations, the column $0\%$ to RSU-only configurations, and the remaining cells to combined RSU and CV
deployments.


\begin{table}
\centering
\caption{RSU deployment positions in meters.}
\label{tab:rsu_deployment}
\setlength{\tabcolsep}{4pt}
\renewcommand{\arraystretch}{0.95}
\begin{tabular}{cll}
\toprule
& HighD / NGSIM & SUMO \\
\midrule
$d_0$ & no RSUs & no RSUs \\
$d_1$ & 380 & 2450 \\
$d_2$ & 20, 380 & 50, 2450 \\
$d_3$ & 20, 220, 380 & 50, 1250, 2450 \\
$d_4$ & 20, 140, 260, 380 & 50, 850, 1650, 2450 \\
$d_5$ & 20, 100, 220, 300, 380 & 50, 650, 1250, 1850, 2450 \\
\bottomrule
\end{tabular}
\end{table}

For each parameter combination, we conduct $100$ independent trials with different random CV subsets. For HighD and NGSIM, CVs are drawn uniformly from the vehicles active during the observation window ($774$ vehicles for HighD and $1668$ for NGSIM). For SUMO, the ego-CV is retained in all CV-based trials, and the remaining CVs are drawn uniformly from the other $248$ vehicles active during its traversal window. In the $0\%$ penetration case, no vehicles are designated as CVs, giving an RSU-only baseline. Since there is no CV-sampling randomness, this case is evaluated in a single run.

Performance is evaluated using the fused network estimate, obtained by averaging the local estimates of the active sensor nodes at each time step
$
\bar{x}_{k,j} =
\frac{1}{|\mathcal{S}_{k,j}|}
\sum_{l \in \mathcal{S}_{k,j}} \hat{x}^{\,l}_{k,j},
$
where $\mathcal{S}_{k,j}$ is the set of active sensor nodes in trial $j$ at
time $k$. This fused estimate measures the collective accuracy of the
distributed network rather than the estimate at a single node.

For each trial $j$ and state variable $z \in \{\rho,v\}$, we first compute
the Root Mean Square Error (RMSE) as
\[
\mathrm{RMSE}_j(z) =
\sqrt{
\frac{1}{N_\mathrm{steps}N}
\sum_{k=1}^{N_\mathrm{steps}}
\|z_{k}-\bar{z}_{k,j}\|^2}
\]
and then the Symmetric Mean
Absolute Percentage Error (SMAPE) as
\[
\mathrm{SMAPE}_j(z) =
\frac{100}{N_\mathrm{steps}N}
\sum_{k=1}^{N_\mathrm{steps}}
\sum_{i=1}^{N}
\frac{2|z_{i,k}-\bar{z}_{i,k,j}|}
{|z_{i,k}|+|\bar{z}_{i,k,j}|}.
\]
The RMSE measures the absolute error in physical units, while the SMAPE measures relative error on a bounded, unit-independent scale.

For the first experiment, we report the full trial distributions of $\mathrm{RMSE}_j$ and $\mathrm{SMAPE}_j$ using boxplots. For the second, each configuration is summarized by the average of the median
$\mathrm{SMAPE}_j(\rho)$ and median $\mathrm{SMAPE}_j(v)$ over the $100$
trials.

\begin{figure}
    \centering
    \includegraphics[width=\linewidth]{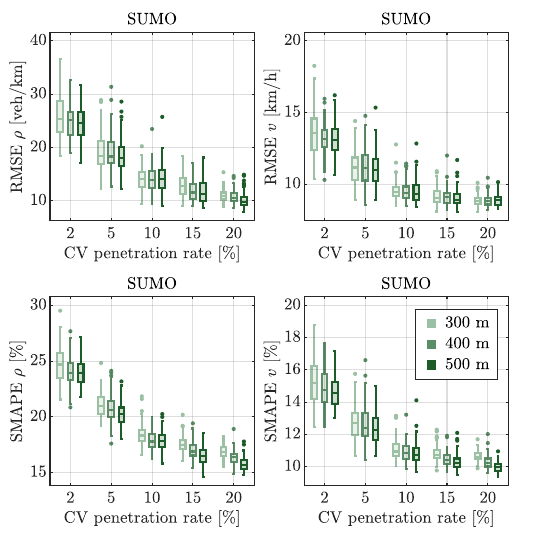}
    \caption{Estimation accuracy of the fused network estimate in SMAPE and RMSE for the SUMO scenario across CV penetration rates and V2X communication ranges, over 100 trials.}
    \label{fig:sumo_boxplot}
\end{figure}

\begin{figure}
    \centering
    \includegraphics[width=\linewidth]{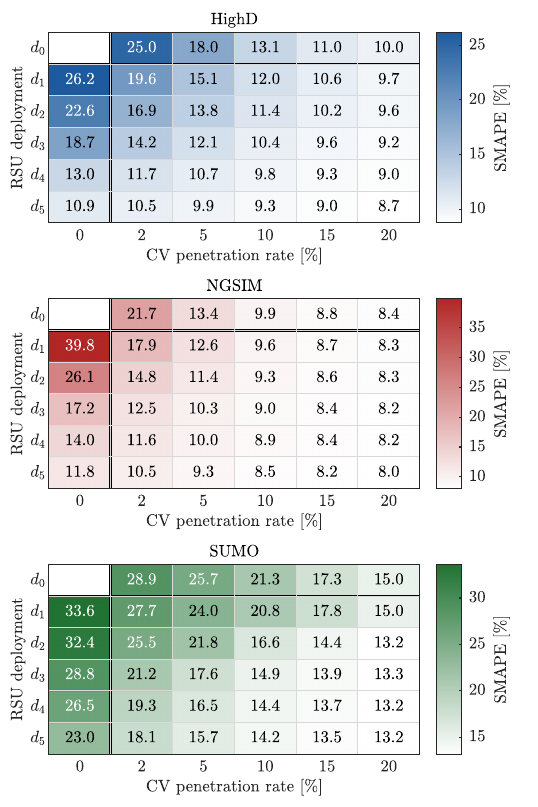}
\caption{Average median $\mathrm{SMAPE}(\rho)$ and $\mathrm{SMAPE}(v)$ of the fused network estimate over 100 trials for varying CV penetration rates (columns) and RSU deployments $d_0,\dots,d_5$ (rows) in HighD, NGSIM, and SUMO. The row $d_0$ corresponds to CV-only configurations, the column $0\%$ to RSU-only configurations, and the remaining cells to combined RSU and CV deployments.}
    \label{fig:RSUvsCV}
\end{figure}

\subsection{HighD and NGSIM: CV Penetration and V2X Range}\label{subsec:stats_empirical}
Figures~\ref{fig:highd_boxplot} and~\ref{fig:ngsim_boxplot} show the HighD
and NGSIM error distributions for varying V2X range
$r \in \{75,150,300\}$~m and CV penetration rate, using the baseline RSU deployment $d_2$.
For both datasets, the median error decreases and the interquartile range narrows as the CV penetration rate increases, confirming that a denser CV fleet improves both the accuracy and robustness of the DTSE estimates for density and speed alike. At the baseline range $r = 300$~m, as the penetration rate increases from $2\%$ to $20\%$, $\mathrm{SMAPE}(\rho)$ decreases from approximately $17\%$ to $8.5\%$ for HighD and from $13\%$ to $4.5\%$ for NGSIM, while $\mathrm{SMAPE}(v)$ decreases from approximately $17\%$ to $10\%$ for HighD and from $17\%$ to $12\%$ for NGSIM. $\mathrm{RMSE}(v)$ decreases to approximately $4.4$~km/h for HighD and
$3.9$~km/h for NGSIM at the highest penetration rate. At the baseline penetration rate of $10\%$, consistent with Section~\ref{sec:estimation_results}, $\mathrm{SMAPE}(\rho)$ is approximately $11\%$ for HighD and $6\%$ for NGSIM, while $\mathrm{SMAPE}(v)$ is approximately $12\%$ for HighD and $13\%$ for NGSIM.

The effect of V2X communication range is more modest than that of CV penetration for HighD, where the $150$ and $300$~m ranges produce nearly indistinguishable error distributions and only the shortest $75$~m range
yields slightly higher density errors. This reflects a saturation pattern: since the two RSUs remain connected through persistent P2P links, range mainly determines how effectively mobile CV measurements reach the consensus network, with shorter ranges giving weaker, more intermittent connectivity and diminishing returns beyond a moderate range. NGSIM deviates from this pattern at short range, where errors are markedly worse across nearly all penetration rates and metrics (e.g., $\mathrm{SMAPE}(\rho)$ at $p = 2\%$ is roughly $40\%$ higher). We attribute this to NGSIM's more severe, spatially uneven congestion (Fig.~\ref{fig:ngsim_rsu1}). Sharp density gradients between congested and free-flowing regions leave some CVs beyond the communication range of their neighbors, fragmenting the network into weakly connected clusters. The resulting isolated nodes degrade the fused network-wide estimate far more than they would affect the estimate of a single, favorably placed reference node.

Conversely, at the highest penetration rate, shorter V2X ranges achieve marginally lower speed errors than the $300$~m baseline, possibly because a densely connected network at high penetration introduces redundancy in overlapping neighbor updates that mildly conflict during the fusion step, although density errors continue to favor larger ranges.

Overall, CV penetration is the dominant driver of estimation accuracy in this setting, while V2X range plays a secondary, saturating role once a moderate communication range is reached. Even at the realistic baseline of $10\%$ CV penetration and a $300$~m range, the DTSE algorithm achieves accurate, robust spatiotemporal reconstruction across both datasets.

\subsection{SUMO: CV Penetration and V2X Range}\label{subsec:stats_sumo}
Figure~\ref{fig:sumo_boxplot} shows the SUMO error distributions for
V2X range $r \in \{300,400,500\}$ m, using the baseline RSU deployment $d_4$. 
As in the empirical datasets, increasing the CV penetration rate reduces both the median error and its variability for density and speed alike. The effect of V2X communication range remains modest, but unlike HighD and NGSIM, where accuracy saturates beyond moderate ranges, the errors here decrease monotonically with range across penetration rates. At the baseline $r = 400$~m, $\mathrm{SMAPE}(\rho)$ decreases from approximately $24\%$ to $16\%$ and $\mathrm{SMAPE}(v)$ from approximately $15\%$ to $10\%$ as $p$ increases from $2\%$ to $20\%$. This weaker saturation, compared to HighD and NGSIM, is likely a consequence of scale. The tested ranges span only $12$-$20\%$ of the $2.5$~km SUMO domain, compared to $19$-$75\%$ of the $400$~m HighD/NGSIM segment, so each additional $100$~m still meaningfully extends the ego-CV's reachable neighborhood. The same scale difference is visible in the spread of the boxplots, whose interquartile ranges and whiskers are noticeably wider than for HighD or NGSIM, since a fixed CV penetration rate admits far more distinct spatial configurations over a $2.5$~km domain than over a $400$~m segment.

The sensitivity to CV penetration is further amplified by the nature of the SUMO scenario. Unlike the recurring, spatially distributed congestion in HighD and NGSIM, the SUMO shockwave is a single, localized event confined to a narrow region of the domain (Fig.~\ref{fig:SUMO_egoCV}), so accurately capturing it depends partly on whether a CV happens to be nearby when it forms. Higher penetration rates make this far more likely, explaining why increasing $p$ has an outsized benefit here relative to the more uniformly congested empirical datasets.

At the realistic baseline of $10\%$ CV penetration and $r = 400$~m, consistent with the scenario in Section~\ref{sec:estimation_results}, the ego-CV reconstructs the shockwave with $\mathrm{SMAPE}(\rho)$ of approximately $18\%$ and $\mathrm{SMAPE}(v)$ of approximately $11\%$, accuracy comparable to the empirical RSU deployments despite the sparser, time-varying V2X connectivity.

\subsection{CV Penetration and RSU Deployment}

Figure~\ref{fig:RSUvsCV} reports the average of the median $\mathrm{SMAPE}(\rho)$ and median $\mathrm{SMAPE}(v)$ over the Monte Carlo trials for varying CV penetration rates and RSU deployments. The row $d_0$ corresponds to CV-only configurations, the column $0\%$ to RSU-only configurations, and the remaining entries to combined RSU and CV deployments. The V2X communication range is fixed at $300$~m for HighD and NGSIM and $400$~m for SUMO, so that the experiment isolates the effects of RSU deployment density and CV penetration rate.

The results show that both sensing modalities improve the reconstruction accuracy, but in different ways. In the CV-only configurations, increasing the penetration rate consistently reduces the SMAPE, since additional CVs provide more spatially distributed measurements along the corridor. Similarly, in the RSU-only configurations, denser RSU deployments reduce the SMAPE by increasing the spatial coverage of the fixed sensing infrastructure. 

At the baseline configurations used in the preceding experiments, the combined RSU and CV deployments outperform the corresponding single-modality cases. For HighD, the baseline configuration gives a SMAPE of $11.4\%$, compared with $22.6\%$ for RSU-only sensing and $13.1\%$ for CV-only sensing. For NGSIM, the corresponding values are $9.3\%$, $26.1\%$, and
$9.9\%$, respectively. For SUMO, the combined configuration gives $14.4\%$, compared with $26.5\%$ for RSU-only sensing and $21.3\%$ for CV-only sensing. Thus, while CV-only sensing already performs well in the empirical datasets at moderate penetration rates, the addition of RSUs still improves the fused estimate, and in SUMO the combined deployment substantially outperforms CV-only sensing. Conversely, RSUs alone approach the combined accuracy only under dense deployments in the empirical datasets, and remain markedly less accurate in SUMO even with five RSUs.

The combined deployments highlight the complementary roles of the two sensing modalities, with CVs providing mobile spatial measurements along the traffic stream and RSUs providing fixed reference measurements and improved consensus connectivity. This benefit is most evident when either modality alone gives incomplete coverage, particularly in the larger SUMO scenario with a localized transient shockwave, where accurate reconstruction depends more strongly on network-wide information exchange.

\section{Conclusion and Future Work}
\label{sec:conclusion}

This paper presented a distributed traffic-state estimation framework that uses connected vehicles and roadside infrastructure as sensing and estimation nodes. It combines an information-form distributed Kalman filter for the second-order ARZ model with consensus-based information exchange and projection onto a physically admissible state space. This allows RSU and CV measurements to be fused locally while maintaining physical consistency and reconstructing non-equilibrium traffic dynamics.

The framework was evaluated on three highway scenarios. HighD and NGSIM provide empirical traffic data over highway segments with recurring stop-and-go traffic. These datasets allowed us to evaluate our estimation algorithm under real-world congested conditions. The SUMO scenario uses a longer controlled highway segment with a temporary downstream bottleneck. This allowed us to evaluate whether an ego-CV can estimate an approaching congestion wave before it enters its onboard sensing range. Under sparse RSU deployments, V2X communication ranges of $300$–$400$~m, and a $10\%$ CV penetration rate, the estimator recovered the spatiotemporal density and speed patterns across these scenarios.

Monte Carlo experiments on HighD, NGSIM, and SUMO quantified the effects
of CV penetration, V2X communication range, and RSU deployment density.
Higher CV penetration consistently improved accuracy, increasing the V2X
range helped until the communication graph was sufficiently connected,
and denser RSU deployments also improved performance. Combined RSU-CV deployments outperformed the corresponding single-modality setups in nearly all tested configurations, and under sparse RSU deployments and realistic V2X conditions they reduced estimation errors by about a factor of two relative to RSU-only sensing. The results highlighted the complementary roles of the two modalities, with mobile CV measurements providing spatial coverage that fixed infrastructure alone attains only under dense deployment.

Future work will consider packet loss, latency, asynchronous updates, and temporary network disconnections. Other extensions include online ARZ parameter calibration and the consideration of larger highway networks with ramps, lane changes, and multiple interacting bottlenecks.



\appendix

\subsection{Proof of Theorem~\ref{thm:error-guarantee}}
\label{appendix:proof-error-guarantee}
    In the fusion steps \eqref{eq:multi-hop-consensus} and \eqref{eq:fused}, let $\Pi_k \in \mathbb{R}^{|\mathcal{S}^*| \times |\mathcal{S}^*|}$ be the doubly stochastic weight matrix corresponding to the communication graph $\mathcal{G}_k$, where the entry $[\Pi_k]_{lj} = \pi_{(l,j),k}$. For $\alpha = 1$, we have
    $
    \Xi_{k|k}^{l,(1)} = \sum_{j \in \mathcal{S}^*} [\Pi_k]_{lj} \Xi_{k|k}^j.
    $
    Applying this linear protocol iteratively for $L$ steps is equivalent to multiplying the initial state by the matrix $\Pi_k$ raised to the power of $L$. Define this effective multi-step consensus weight matrix as $W_k = (\Pi_k)^L $. 
    Since $\Pi_k$ is doubly stochastic, $W_k$ is also doubly stochastic.
    After $L$ iterations, the final fused information matrix $\overline{\Xi}_{k|k}^l$ for node $l$ is given by 
    $$ 
    \overline{\Xi}_{k|k}^l = \sum_{j \in \mathcal{S}^*} [W_k]_{lj} \Xi_{k|k}^j = \sum_{j \in \mathcal{S}^*} [W_k]_{lj} (\Xi_{k|k-1}^j + \Theta_k^j)
    $$
    where we used \eqref{eq:local-update}.
    Because the matrix inversion is strictly convex on the cone of positive definite matrices, Jensen's inequality and the fact that $\Theta_k^j= (C_k^j)^\T (R_k^j)^{-1} C_k^j\succeq 0$ imply
    \begin{equation}
        \label{eq:bar-Xi-ineq}
        (\overline{\Xi}_{k|k}^l)^{-1} \preceq \sum_{j \in \mathcal{S}^*} [W_k]_{lj} (\Xi_{k|k}^j)^{-1} .
    \end{equation}
    By the Woodbury matrix identity, the local prior information update \eqref{eq:pred-cov} can be rewritten as
    \begin{equation}
        \label{eq:woodbury-expansion}
        \Xi_{k|k-1}^j = \left( \Lambda_{k-1}^j (\overline{\Xi}_{k-1|k-1}^j)^{-1} (\Lambda_{k-1}^j)^\top + Q \right)^{-1} 
    \end{equation}
    Because the local state estimates are random variables, their corresponding Jacobians $\Lambda_{k-1}^l$ fluctuate stochastically. To establish a deterministic baseline for observability, we substitute the estimated Jacobians with $\Lambda_{k-1}$ in \eqref{eq:Lambda_k}. This substitution is justified by the principles of local stability: because the ARZ flux and pressure functions are continuously differentiable, the discrepancy $(\Lambda_{k-1} - \Lambda_{k-1}^l)$ is bounded proportionally to the estimation error. The perturbation introduced by evaluating the system along the true trajectory is of second order and is entirely absorbed by the Taylor series remainder bound \eqref{eq:lin-remainder-bound} within the filter's local region of attraction.

    Because $\overline{\Xi}_{k-1|k-1}^l \succ 0$ and the discretized transition matrix $\Lambda_{k-1}$ is non-singular (guaranteed by the CFL stability condition of the ARZ spatial discretization), the matrix $\Lambda_{k-1} (\overline{\Xi}_{k-1|k-1}^l)^{-1} \Lambda_{k-1}^\top$ is symmetric positive definite. Furthermore, the process noise covariance $Q \succ 0$. Because the $\Lambda_k$ and $Q$ are uniformly bounded, there exists a constant $c > 0$ such that 
    $ 
    Q \preceq c \left( \Lambda_{k-1} (\overline{\Xi}_{k-1|k-1}^l)^{-1} \Lambda_{k-1}^\top \right).
    $
    Substituting this into \eqref{eq:woodbury-expansion} and reversing the Löwner partial order via matrix inversion yields a uniform decay factor $\delta = \frac{1}{1+c} \in (0, 1)$, giving
    $
    \Xi_{k|k-1}^l \succeq \delta \Lambda_{k-1}^{-\top} \overline{\Xi}_{k-1|k-1}^l \Lambda_{k-1}^{-1}.
    $
    Substituting this inequality into \eqref{eq:bar-Xi-ineq} and using \eqref{eq:local-update}, we obtain
    $$ 
    \overline{\Xi}_{k|k}^l \succeq \sum_{j \in \mathcal{S}^*} [W_k]_{lj} \Theta_k^j + \delta \sum_{j \in \mathcal{S}^*} [W_k]_{lj} \Lambda_{k-1}^{-\top} \overline{\Xi}_{k-1|k-1}^j \Lambda_{k-1}^{-1}.
    $$
    Define the multi-step effective weight matrix $\mathcal{W}_{t,k} = W_k W_{k-1} \dots W_{t+1}$ (with $\mathcal{W}_{k,k} = W_k$). Unrolling the above recursive inequality backward in time over a window $K = \max(b, N_o)$, where $b$ and $N_o$ are defined in Definitions~\ref{def:joint-connectivity} and \ref{def:obs}, we obtain
    $$ 
    \overline{\Xi}_{k|k}^l \succeq \sum_{t=k-K}^{k} \delta^{k-t} \sum_{j \in \mathcal{S}^*} [\mathcal{W}_{t,k}]_{lj} (\Phi_{t, k})^\top \Theta_t^j \Phi_{t, k}.
    $$
    Since the communication graph $\mathcal{G}_k$ is jointly connected over the window $b$ (and consequently over $K \geq b$), the product of doubly stochastic matrices over a jointly connected sequence guarantees that the resulting multi-step weight matrix $\mathcal W_{t,k}$ is element-wise positive. That is, there exists $\underline{w}  > 0$ such that $[\mathcal{W}_{t,k}]_{lj} \geq \underline{w} $ for all nodes $l, j \in \mathcal{S}^*$ and all times $t \in [k-K, k]$. Factoring the worst-case decay $\delta^K$ and the connectivity bound $\underline{w} $ out of the summations gives
    \begin{align*}
        \overline{\Xi}_{k|k}^l &\succeq \underline{w} \delta^K \sum_{t=k-K}^{k} (\Phi_{t, k})^\top \left( \sum_{j \in \mathcal{S}^*} (C_t^j)^\top (R_t^j)^{-1} C_t^j \right) \Phi_{t, k} \\
        &= \underline{w} \delta^K \sum_{t=k-K}^{k} (\Phi_{t, k})^\top C_t^\top R_t^{-1} C_t \Phi_{t, k} = \underline{w}  \delta^K \mathcal{O}_{k, K}.
    \end{align*}
    By the collective observability assumption, $\mathcal{O}_{k, K} \succeq \underline{o} I_{2N}$. Setting $\underline{\sigma} = \underline{w} \delta^K \underline{o} > 0$, we conclude that $\overline{\Xi}_{k|k}^l \succeq \underline{\sigma} I_{2N}$. Furthermore, because $Q \succ 0$ and $R_k^l \succ 0$ are bounded, \eqref{eq:local-update} and \eqref{eq:woodbury-expansion} imply $\overline{\Xi}_{k|k}^l \preceq \overline{\sigma} I_{2N}$ for some $\overline{\sigma}\geq \underline{\sigma}$.

    To prove the exponential mean-square boundedness of the estimation error, let $e_k^l = x_k - \hat{x}_k^l$ be the projected error at node $l$. The updated prior information vector is $\xi_{k|k-1}^l = \Xi_{k|k-1}^l \hat{x}_k^l$ and, following the similar arguments as with $\overline{\Xi}_{k|k}^l$, we have
    $$
    \overline{\xi}_{k|k}^l = \sum_{j \in \mathcal{S}^*} [W_k]_{lj} \left( \Xi_{k|k-1}^j \hat{x}_k^j + (C_k^j)^\top (R_k^j)^{-1} y_k^j \right).
    $$
    Let $\overline{x}_{k|k}^l = (\overline{\Xi}_{k|k}^l)^{-1} \overline{\xi}_{k|k}^l$ define the fused posterior estimate, and $\overline{e}_{k|k}^l = x_k - \overline{x}_{k|k}^l$ be the fused posterior error. By expanding the true state as $x_k = (\overline{\Xi}_{k|k}^l)^{-1} \sum_j [W_k]_{lj} (\Xi_{k|k-1}^j + \Theta_k^j) x_k$, and substituting $y_k^j = C_k^j x_k + \nu_k^j$, we subtract the estimate from the true state. The $\Theta_k^j x_k$ terms cancel, yielding
    $$
    \overline{e}_{k|k}^l = (\overline{\Xi}_{k|k}^l)^{-1} \sum_{j \in \mathcal{S}^*} [W_k]_{lj} \left[ \Xi_{k|k-1}^j e_k^j - (C_k^j)^\top (R_k^j)^{-1} \nu_k^j \right].
    $$
    The unprojected prior estimate for the next time step is 
    $$
    \check{x}_{k+1|k}^l = (\Xi_{k+1|k}^l)^{-1} \xi_{k+1|k}^l = \Lambda_k^l \overline{x}_{k|k}^l + \eta_k^l.
    $$
    Substituting $\eta_k^l$ and $\Lambda_k^l$ from \eqref{eq:linearized_dyn}, we obtain
    $$
    \check{x}_{k+1|k}^l = A \overline{x}_{k|k}^l + G f(\overline{x}_{k|k}^l, u_k).
    $$
    Let $\check{e}_{k+1|k}^l = x_{k+1} - \check{x}_{k+1|k}^l$ be the unprojected prior error. Subtracting the prediction from the true state \eqref{eq:ssARZ} yields
    $$
    \check{e}_{k+1|k}^l = A(x_k - \overline{x}_{k|k}^l) + G \left( f(x_k, u_k) - f(\overline{x}_{k|k}^l, u_k) \right) + \omega_k.
    $$
    By adding and subtracting $G \nabla f(\overline{x}_{k|k}^l) \overline{e}_{k|k}^l$, we obtain
    \begin{equation}
        \label{eq:prior-error}
        \check{e}_{k+1|k}^l = \Lambda_k^l \overline{e}_{k|k}^l + \varphi_k^l + \omega_k
    \end{equation}
    where $\varphi_k^l = \varphi(x_k,\overline{x}_{k|k}^l)$. By Assumption~\ref{assume:bounded-nonlin-residual}, its norm satisfies $\|\varphi_k^l\| \leq \varrho \|\overline{e}_{k|k}^l\|^2$. 
    To establish stability, we analyze the prediction matrix update \eqref{eq:pred-cov}, which, by the Woodbury matrix identity, is equivalent to
    $$
    \Xi_{k+1|k}^l = \left( Q + \Lambda_k^l (\overline{\Xi}_{k|k}^l)^{-1} (\Lambda_k^l)^\top \right)^{-1}.
    $$
    Because $\overline{\Xi}_{k|k}^l \succeq \underline{\sigma} I$ and the system state is physically bounded, $\Lambda_k^l (\overline{\Xi}_{k|k}^l)^{-1} (\Lambda_k^l)^\top$ is bounded from above. Thus, there exists a strictly positive scalar $\vartheta > 0$ such that $Q \succeq \vartheta \Lambda_k^l (\overline{\Xi}_{k|k}^l)^{-1} (\Lambda_k^l)^\top$, which implies $(Q + \Lambda_k^l (\overline{\Xi}_{k|k}^l)^{-1} (\Lambda_k^l)^\top) \succeq (1+\vartheta) \Lambda_k^l (\overline{\Xi}_{k|k}^l)^{-1} (\Lambda_k^l)^\top$ and yields
    $$
    \Xi_{k+1|k}^l \preceq \frac{1}{1+\vartheta} (\Lambda_k^l)^{-\top} \overline{\Xi}_{k|k}^l (\Lambda_k^l)^{-1}.
    $$
    Defining the contraction parameter $\zeta \in (0, 1)$ such that $1-\zeta = (1+\vartheta)^{-1}$, we obtain the prediction contraction bound as
    \begin{equation}
        \label{eq:pred-contraction}
        (\Lambda_k^l)^\top \Xi_{k+1|k}^l \Lambda_k^l \preceq (1 - \zeta) \overline{\Xi}_{k|k}^l.
    \end{equation}
    Let $c_\omega$ and $c_\nu$ be such that
    \begin{align*}
        \sup_{k \geq 0} \sum_{j \in \mathcal{S}^*} \text{Tr}\left(\Xi_{k+1|k}^j Q\right) & \leq c_\omega \\
        \sup_{k \geq 0} \sum_{l \in \mathcal{S}^*} \mathbb{E} \left[ (\tilde{\nu}_k^l)^\top \overline{\Xi}_{k|k}^l \tilde{\nu}_k^l \right] &\leq c_\nu
    \end{align*}
    where $\tilde{\nu}_k^l = (\overline{\Xi}_{k|k}^l)^{-1} \sum_{j \in \mathcal{S}^*} [W_k]_{lj} (C_k^j)^\top (R_k^j)^{-1} \nu_k^j$ is the effective measurement noise at node $l$ after consensus. 
    
    Define the global stochastic Lyapunov function $$V_{k+1} = \sum_{l \in \mathcal{S}^*} (\overline{e}_{k+1|k+1}^l)^\top \overline{\Xi}_{k+1|k+1}^l \overline{e}_{k+1|k+1}^l.$$
    Recall a standard fact for doubly stochastic weights, $( \sum w_i z_i )^\top ( \sum w_i M_i )^{-1} ( \sum w_i z_i ) \leq \sum w_i z_i^\top M_i^{-1} z_i$, for some $z_i$ and $M_i$. Recall another fact that $\Xi_{k+1|k}^j (\Xi_{k+1|k}^j + \Theta_{k+1}^j)^{-1} \Xi_{k+1|k}^j \preceq \Xi_{k+1|k}^j$.
    Therefore, we obtain
    \begin{align*}
    &\mathbb{E}[V_{k+1}] \\
    &
    \leq \mathbb{E} \bigg[ \sum_{l \in \mathcal{S}^*} \sum_{j \in \mathcal{S}^*} [W_{k+1}]_{lj} (e_{k+1}^j)^\top \Xi_{k+1|k}^j (\Xi_{k+1|k}^j \\
    & \qquad\qquad\qquad\qquad + \Theta_{k+1}^j)^{-1} \Xi_{k+1|k}^j e_{k+1}^j \bigg] + c_{\nu} \\
    & \leq \mathbb{E} \bigg[ \sum_{j \in \mathcal{S}^*} (e_{k+1}^j)^\top \Xi_{k+1|k}^j e_{k+1}^j \bigg] + c_{\nu}.
    \end{align*}
    The state estimate $\hat{x}_{k+1}^j = \Pi_{\mathcal{D}}(\check{x}_{k+1|k}^j)$. Because the true state satisfies $x_{k+1} \in \mathcal{D}$, and $\mathcal{D}$ is a hyperrectangle aligned with the state axes (bounds on $\rho$ and $\psi$), the projection $\Pi_{\mathcal{D}}$ is non-expansive. 
    Given the diagonal dominance of the information matrices, the projection preserves the weighted quadratic norm: $(e_{k+1}^j)^\top \Xi_{k+1|k}^j e_{k+1}^j \leq (\check{e}_{k+1|k}^j)^\top \Xi_{k+1|k}^j \check{e}_{k+1|k}^j$.
    Substituting \eqref{eq:prior-error} into the above bound, and noting that the cross-terms involving the zero-mean process noise $\omega_k$ vanish strictly under expectation, yields
    \begin{multline*}
        \mathbb{E}[V_{k+1}] \leq \mathbb{E} \left[ \sum_{j \in \mathcal{S}^*} (\overline{e}_{k|k}^j)^\top (\Lambda_k^j)^\top \Xi_{k+1|k}^j \Lambda_k^j \overline{e}_{k|k}^j \right] \\ 
        + \mathbb{E} \left[ \sum_{j \in \mathcal{S}^*} \Psi(\varphi_k^j) \right] + c_{\nu} + c_{\omega}
    \end{multline*}
    where $\Psi(\varphi_k^j) = (\varphi_k^j)^\top \Xi_{k+1|k}^j \varphi_k^j + 2 (\overline{e}_{k|k}^j)^\top (\Lambda_k^j)^\top \Xi_{k+1|k}^j \varphi_k^j$.
    From Assumption~\ref{assume:bounded-nonlin-residual}, we have
    $$
    \sum_{j \in \mathcal{S}^*} \Psi(\varphi_k^j) \leq c_2 \sum_{j \in \mathcal{S}^*} \|\overline{e}_{k|k}^j\|^3 + c_3 \sum_{j \in \mathcal{S}^*} \|\overline{e}_{k|k}^j\|^4 
    $$
    for some $c_2,c_3>0$ that depend on $\varrho$ in \eqref{eq:lin-remainder-bound}.
    Applying \eqref{eq:pred-contraction} reduces the nominal term to $(1-\zeta) V_k$. The noise traces aggregate into a uniform constant $c' = c_{\nu} + c_{\omega}$. Provided the initial error is within a small region of attraction, these terms are absorbed into the linear decay rate, yielding a strict factor $\lambda \in (1-\zeta, 1)$ such that
    $
    \mathbb{E}[V_{k+1}] \leq \lambda \mathbb{E}[V_k] + c'.
    $
    Unrolling this linear difference inequality from $k=0$ gives the geometric series bound:
    $
    \mathbb{E}[V_k] \leq \lambda^k V_0 + \frac{c'}{1-\lambda}.
    $
    By applying the Rayleigh bounds ($\underline{\sigma} \sum_j \|\overline{e}_{k|k}^j\|^2 \leq V_k \leq \overline{\sigma} \sum_j \|\overline{e}_{k|k}^j\|^2$), we translate the Lyapunov bound back to the expected squared local error norm for an individual node $l$:
    \begin{align*}
    \mathbb{E}[\|\overline{e}_{k|k}^l\|^2] & \leq \frac{1}{\underline{\sigma}} \mathbb{E}[V_k] \\
    & \leq \left(\frac{\overline{\sigma}}{\underline{\sigma}}\right) \left( \sum_{j \in \mathcal{S}^*} \|x_0 - \hat{x}_0^j\|^2 \right) \lambda^k + \left(\frac{c'}{\underline{\sigma}(1-\lambda)}\right).
    \end{align*}
    Assuming uniform initialization across all nodes ($\hat{x}_0^j = \hat{x}_0^l, \forall j \in \mathcal{S}^*$), the summation over $j$ simplifies to $|\mathcal{S}^*| \|x_0 - \hat{x}_0^l\|^2$. Setting $c_1 = \frac{|\mathcal{S}^*|\overline{\sigma}}{\underline{\sigma}}$ and $M = \frac{c'}{\underline{\sigma}(1-\lambda)}$ verifies \eqref{eq:error-guarantee}.
    Because $c'$ is derived solely from the bounded covariance traces of the process and measurement noise, and $\underline{\sigma}$ is determined by network observability, the asymptotic bound $M$ is structurally independent of the physical density limit $\rho_m$. 
    This concludes the proof.

\bibliographystyle{IEEEtran}
\bibliography{biblio}    

\begin{IEEEbiography}[{\includegraphics[width=1in,height=1.25in,clip,keepaspectratio]{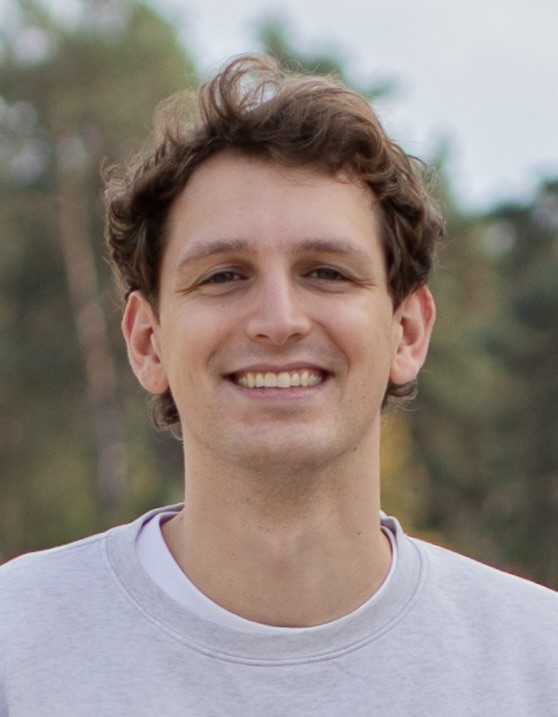}}]{Vincent de Heij}
is a Ph.D. student in Systems and Control at the University of Groningen, The Netherlands. He received the B.Sc. degree in Industrial Engineering and Management and the M.Sc. degree in Systems and Control from the University of Groningen. During his graduate studies, he was a research intern at KTH Royal Institute of Technology, Stockholm, Sweden. His research interests include distributed estimation and control, multi-agent formation control, and intelligent transportation systems.
\end{IEEEbiography}

\begin{IEEEbiography}[{\includegraphics[width=1in,height=1.25in,clip,keepaspectratio]{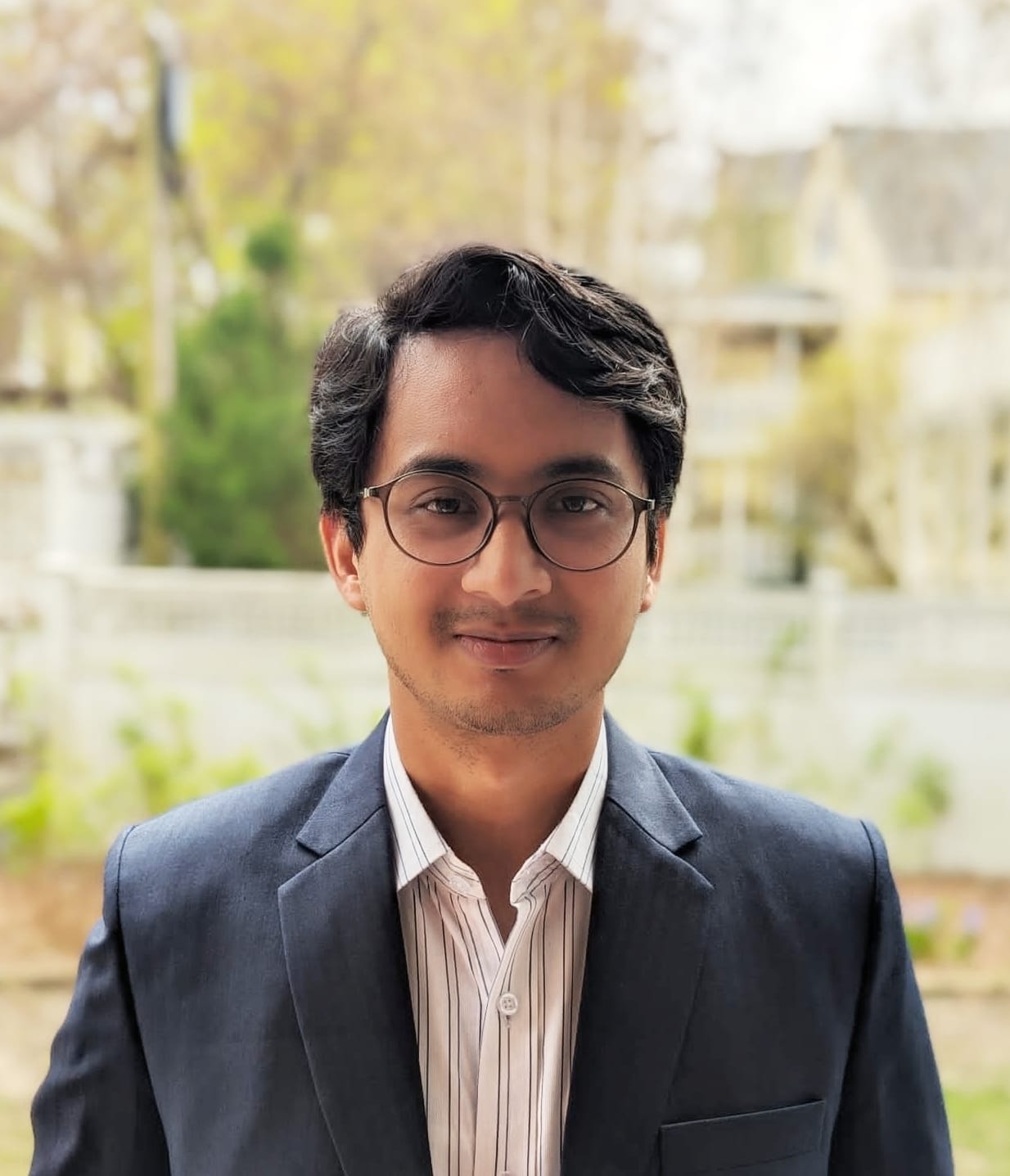}}]{M. Umar B. Niazi}
is a research scientist at the KTH Royal Institute of Technology, Sweden. 
Prior to that, he was a postdoctoral researcher at the Massachusetts Institute of Technology, USA, and KTH Royal Institute of Technology, Sweden. 
He received the MS degree in Electrical and Electronics Engineering from Bilkent University, Turkey, and the PhD in Automatic Control from Grenoble INP, Université Grenoble Alpes, France, where he was affiliated with INRIA Grenoble and GIPSA-Lab as a graduate research assistant.
His research focuses on the security, resilience, and optimization of socio-cyber-physical systems, with particular emphasis on intelligent transportation networks and urban mobility.
He is a recipient of the Marie-Curie postdoctoral fellowship from the European Commission. He was a finalist for the Best Student Paper Award at the 2019 European Control Conference.
\end{IEEEbiography}

\begin{IEEEbiography}[{\includegraphics[width=1in,height=1.25in,clip,keepaspectratio]{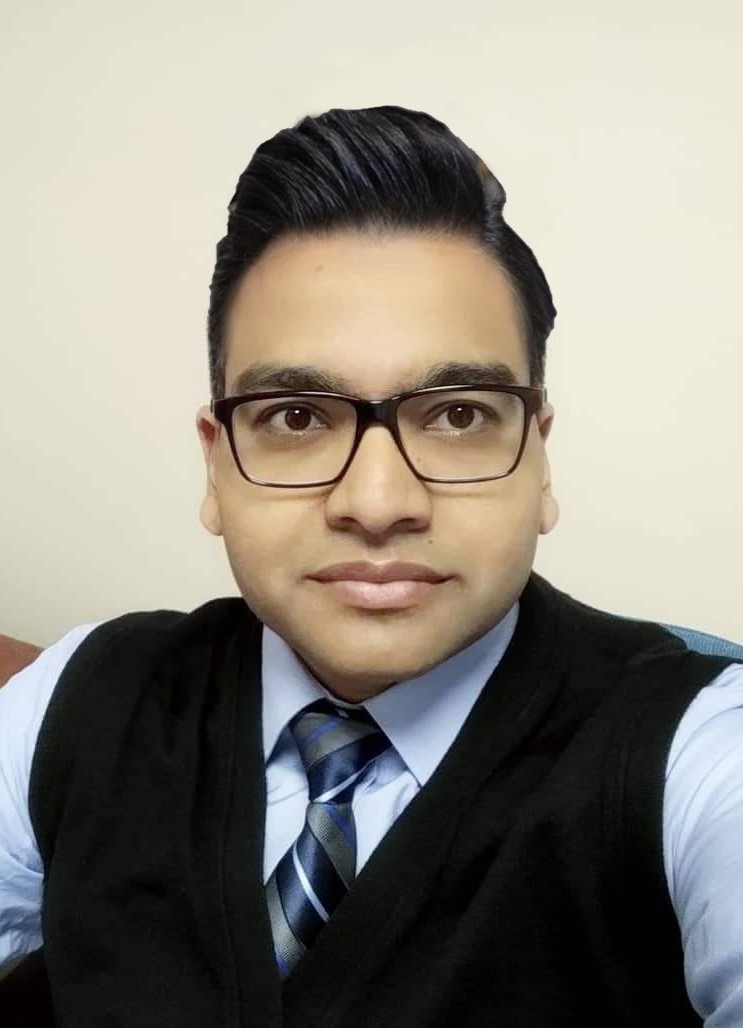}}]{Saeed Ahmed}  is a tenured Assistant Professor of Systems and Control at the University of Groningen, The Netherlands, where he is affiliated with the Engineering and Technology Institute Groningen and the Jan C. Willems Center for Systems and Control. Prior to joining this position, he held postdoctoral positions at the University of Groningen and at the Technical University of Kaiserslautern (now RPTU), Germany. He completed his Ph.D. at Bilkent University, Turkey, in collaboration with Inria, CentraleSupélec, University of Paris-Saclay, France.  His research interests span various topics in systems and control engineering. From a theoretical point of view, he is interested in stability and control, feedback optimization, observer design, and nonlinear and hybrid (switched and impulsive) systems. From an application point of view, he is interested in designing intelligent control algorithms for autonomous vehicles and energy systems.  He received the best presentation award in the Control/Robotics/Communications/Network category at the IEEE Graduate Research Conference 2018 held in Bilkent University, Turkey, the outstanding reviewer award from the European Journal of Control in 2017, and the Investments in Practical Innovations (IPI) Award 2025 from the University of Groningen, The Netherlands. He was the supervisor of the Champion team of the National RDW Self-Driving team 2026. He is an associate editor of IEEE Transactions on Automatic Control (starting August 2026), Systems and Control Letters, and the IEEE Technology Conference Editorial Board (TCEB). He is a member of the IFAC Technical Committees on Networked Systems, Non-linear Control Systems, and Distributed Parameter Systems. 
\end{IEEEbiography}

\vspace{-0.5em}

\begin{IEEEbiography}[{\includegraphics[width=1in,height=1.25in,clip,keepaspectratio]{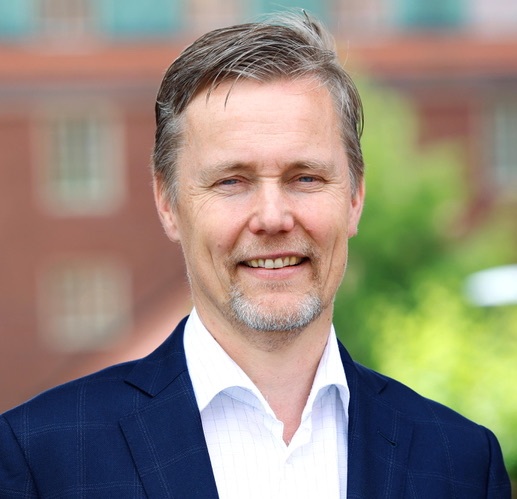}}]{Karl Henrik Johansson}
is Swedish Research Council Distinguished Professor in Electrical
Engineering and Computer Science at KTH Royal
Institute of Technology in Sweden and Founding
Director of Digital Futures. He earned his MSc degree in Electrical Engineering and PhD in Automatic
Control from Lund University. He has held visiting
positions at UC Berkeley, Caltech, NTU and other
prestigious institutions. His research interests focus
on networked control systems and cyber-physical
systems with applications in transportation, energy,
and automation networks. For his scientific contributions, he has received
numerous best paper awards and various other distinctions from IEEE, IFAC,
and other organizations. He has been awarded Distinguished Professor by
the Swedish Research Council, Wallenberg Scholar by the Knut and Alice
Wallenberg Foundation, Future Research Leader by the Swedish Foundation
for Strategic Research. He has also received the triennial IFAC Young Author
Prize, IEEE CSS Distinguished Lecturer, IFAC Outstanding Service Award,
and IEEE CSS Hendrik W. Bode Lecture Prize. His extensive service to
the academic community includes being President of the European Control
Association, IEEE CSS Vice President Diversity, Outreach and Development,
and Member of IEEE CSS Board of Governors and IFAC Council. He has
served on the editorial boards of Automatica, IEEE TAC, IEEE TCNS and
many other journals. He has also been a member of the Swedish Scientific
Council for Natural Sciences and Engineering Sciences. He is Fellow of both
the IEEE and the Royal Swedish Academy of Engineering Sciences.
\end{IEEEbiography}

\end{document}